\documentclass[12pt]{article}
\usepackage{amsmath,amssymb,amsfonts,amscd,amsxtra,amsthm}
\usepackage{tikz-cd}
\usepackage{tikz}
\usepackage{latexsym}
\usepackage{xcolor}
\usepackage{epstopdf}
\usepackage{graphicx}
\usepackage{epstopdf}
\usepackage{subfig}
\usepackage{float}
\usepackage{overpic}

\usepackage{algorithmic}
\usepackage[ruled,linesnumbered,noend]{algorithm2e}
\usepackage{booktabs}
\usepackage{appendix}
\usepackage[top=1in, bottom=1in, left=1in, right=1in]{geometry}
\renewcommand{\theequation}{\thesection.\arabic{equation}}

\newtheorem{thm}{Theorem}[section]

\newtheorem{prop}[thm]{Proposition}
\newtheorem{defn}{Definition}

\newtheorem{cond}{Condition}

\newtheorem{rmk}[thm]{Remark}
\newtheorem{eg}{Example}

\renewcommand{\Re}{{\mbox{Re}}}
\newcommand{\Real}{\mathbb R}
\newcommand{\comp}{\mathbb C}

\newcommand{\Ord}{\mathcal{O}}

\newcommand{\PL}{\mathcal{P}_L}
\newcommand{\PH}{\mathcal{P}_H}

\newcommand{\HL}{\mathcal{H}_L}
\newcommand{\HH}{\mathcal{H}_H}

\newcommand{\bra}[1]{\left(#1\right)}

\renewcommand\appendix{\par
  \setcounter{section}{0}
  \setcounter{subsection}{0}
  \setcounter{figure}{0}
  \setcounter{table}{0}
  \renewcommand\thesection{Appendix \Alph{section}}
  \renewcommand\theequation{\Alph{section}.\arabic{equation}}
  \renewcommand\thefigure{\Alph{section}.\arabic{figure}}
  \renewcommand\thetable{\Alph{section}.\arabic{table}}
  \renewcommand\thethm{\Alph{section}.\arabic{thm}}
}

\numberwithin{equation}{section}

\title{Towards a Theory of Stable Super-Resolution: Model-Based Formulation and Stability Analysis}
 \author{
Zetao Fei
\thanks{Department of Mathematics, 
HKUST,  Clear Water Bay, Kowloon, Hong Kong (zfei@connect.ust.hk).}
and Hai Zhang
\thanks{Department of Mathematics, 
HKUST,  Clear Water Bay, Kowloon, Hong Kong (haizhang@ust.hk). H. Zhang was partially supported by Hong Kong RGC grant GRF 16304621 and NSFC grant 12371425.}
}

\begin{document}

\maketitle

\begin{abstract}
In mathematics, a super-resolution problem can be formulated as acquiring high-frequency data from low-frequency measurements. This extrapolation problem in the frequency domain is well-known to be unstable. We propose a model-based super-resolution framework (Model-SR) for solving the super-resolution problem and analyzing its stability, aiming to narrow the gap between limited theory and the broad empirical success of super-resolution methods. The key rationale is that, to be determined by its low-frequency components, the target signal must possess a low-dimensional structure. Instead of assuming that the signal itself lies on a low-dimensional manifold in the signal space, we assume that it is generated from a model with a low-dimensional parameter space. This shift of perspective allows us to analyze stability directly through the model parameters. Within this framework, we can recover the signal by solving a nonlinear least square problem and achieve super-resolution by extracting its high-frequency components. Theoretically, the resolution-enhancing map is proven to have Lipschitz continuity, with a constant that depends crucially on parameter separation conditions; consequently, measurements generated by well-separated parameters yield stable reconstructions. This separation condition can be effectively enforced via sparsity modeling, which requires using the minimal number of parameters to represent the measured signal, thereby highlighting the role of sparsity in the stability of super-resolution. Moreover, the Lipschitz constant grows with the high-frequency cutoff, ultimately rendering extrapolation ineffective beyond a certain threshold. We apply the general theory to three concrete models and give the stability estimates for each model. Numerical experiments are conducted to show the super-resolution behavior of the proposed framework. The model-based mathematical framework can be extended to problems with similar structures. 

\bigskip

Keywords: super-resolution, stability, optimization, sparsity. 

\bigskip

MSC: 94-10, 65Z05  

\end{abstract}

\section{Introduction}{\label{sec: introduction}}
Appearing in different literature, super-resolution mainly refers to the techniques that enhance the resolution of signals or images. Since the birth of the microscope, super-resolution has been a central problem for imaging systems for about three centuries. In wave-based imaging systems, the resolution is limited due to the diffraction nature of the wave. The resolution limit can be characterized by the Rayleigh length and depends on the cutoff frequency of the system. Super-resolution techniques are therefore widely desired in imaging-related fields such as geophysics \cite{khaidukov2004diffraction}, medical imaging \cite{greenspan2009super}, radar imaging \cite{odendaal1994two}, microscopy \cite{rust2006sub}, etc. 

In mathematical literature, super-resolution usually refers to the stable recovery of high-frequency information from low-frequency measurements. For a general compactly supported function, its Fourier transform is analytic, and hence the extrapolation problem is uniquely solvable in the absence of noise. However, it is notoriously ill-posed in practice, as even small noise leads to severe instability \cite{Vessella1999ACD,demanet2019stable,isaev2020h}. Stable recovery is nevertheless possible when the signal admits additional structure (e.g., point sources, piecewise constant profiles). Another commonly used formulation, especially in image processing, is to reconstruct a high-resolution signal from its low-resolution counterpart. Since the Fourier transform connects the physical and frequency domains, these two viewpoints are essentially equivalent. In what follows, we adopt the frequency-domain formulation and present a general framework for analyzing and solving super-resolution, with an emphasis on stability. For clarity, we focus on the one-dimensional case.

\subsection{Problem outline in the frequency domain}
Let $h$ be the sampling step size in the frequency domain, and $\omega_k = kh$, for $k\in\mathbb Z$. Denote the low-resolution sampling points as $\{\omega_k\}_{k=-K_L}^{K_L}$ and high-resolution sampling points as $\{\omega_k\}_{k=-K_H}^{K_H}$, where $K_L$ and $K_H$ denote the low- and high-frequency cutoffs, respectively. Assume $K_H>K_L$ and define the super-resolution factor (SRF) as 
\begin{align}
    \operatorname{SRF}:= \frac{K_H}{K_L}.
\end{align}

Let $\mathcal{M}\subset \mathcal{S}'(\Real)$ be the signal space in the physical domain, where $\mathcal{S}'(\Real)$ denotes the space of tempered distributions. The sampling in the frequency domain can be written as follows. For $\psi\in \mathcal{M}$, we define the low-resolution sampling operator $G_L: \mathcal{M}\rightarrow\comp^{2K_L+1}$ as
\begin{align}\label{def: GL}
    G_L(\psi) = (g_{-K_L}(\psi),g_{-K_L+1}(\psi),\cdots,g_{K_L}(\psi)),
\end{align}
with 
\begin{align}
    g_{k}(\psi) := \mathcal{F}[\psi](\omega_k) = \int_{\Real} \psi(x)e^{-2\pi i \omega_k x}dx.
\end{align}
Similarly, we define the high-resolution sampling operator $G_H:\mathcal{M}\rightarrow \comp^{2K_H+1}$ as 
    \begin{align}\label{def: GH}
        G_H(\psi) = (g_{-K_H}(\psi),g_{-K_H+1}(\psi),\cdots,g_{K_H}(\psi)).
    \end{align}

We assume that $G_L$ and $G_H$ are continuous.
For $\psi,\tilde{\psi}\in \mathcal{M}$, We say $\psi$ is sampling equivalent to $\tilde{\psi}$ if $g_k(\psi)=g_k(\tilde{\psi})$ for all $k\in\mathbb Z$. We define the signal space as an equivalent class, $[\mathcal{M}]:= \mathcal{M}/_{\sim}$. The motivation for such a definition can be seen in the example below.

\begin{eg}\label{eg: 1}
    Let $\omega_k=k$, $\theta\in[0,1]$, $\psi(x)=\delta_{\theta}$, and $\tilde{\psi}(x) = \delta_{1+\theta}$. Then, we can calculate that $g_k(\psi)=g_k(\tilde{\psi})$ for all $k\in\mathbb Z$ due to the fixed sampling step size. Thus $\psi$ and $\tilde{\psi}$ are in the same equivalent class and they are viewed as the same signal. 
    
\end{eg}

For a slight abuse of notation, we use the notation $\mathcal{M}$ for signal space from now on.

We define the low- and high-resolution signal space as $\HL := G_L(\mathcal{M})$ and $\HH:=G_H(\mathcal{M})$ respectively. 
Let $\mathcal{Q}:\HH\rightarrow\HL$ be the downsampling operator satisfying 
\begin{align}\label{commu:Q}
    \mathcal{Q}\circ G_H = G_L.
\end{align}
The diagram shown in Figure \ref{fig: signal model} commutes.
\begin{figure}[ht]
    \centering    
    \begin{tikzcd}[row sep=small,column sep = large]
                                     &\HL \arrow[from=dd, "\mathcal{Q}"]\\ 
        \mathcal{M}\arrow[ur, "G_L"] \arrow[dr, "G_H",swap]&                   \\ 
                                    &\HH 
    \end{tikzcd}
    \caption{Signal space, low- and high-resolution spaces, and related maps.}
    \label{fig: signal model}
\end{figure}

Generally speaking, super-resolution aims to find a resolution-enhancing map, $\mathcal{L}:\HL\rightarrow\HH$, satisfying the following condition:
\begin{align}\label{commu:L}
    \mathcal{L}\circ G_L = G_H.
\end{align}
Combining (\ref{commu:Q}) and (\ref{commu:L}), we have $(\mathcal{L}\mathcal{Q})\circ G_H=G_H$, which implies that super-resolution essentially aims to find a generalized left inverse of the downsampling operator $\mathcal{Q}$. 

To ensure the uniqueness of super-resolution in the absence of noise, we propose the following condition:
\begin{cond}\label{cond: injective}
    The low-resolution sampling operator $G_L$ is injective.
\end{cond}

The condition above ensures that the signal profile is identifiable from the measurements. Otherwise, if there exist distinct profiles $\psi \neq \tilde{\psi} \in \mathcal{M}$ with $G_L(\psi) = G_L(\tilde{\psi})$, then for some integer $K_H$ we must have $G_H(\psi) \neq G_H(\tilde{\psi})$. In that case, the resolution-enhancing map $\mathcal{L}$ is not well-defined, and the super-resolution problem becomes non-unique in the absence of Condition \ref{cond: injective}. To ensure uniqueness and identifiability, we assume throughout that
\[
2K_L + 1 > \dim(\mathcal{M}),
\]
where $\dim(\mathcal{M})$ denotes the intrinsic degrees of freedom of the signal space $\mathcal{M}$ (see formal definition in Definition \ref{def: dim M}). This condition asserts that $\mathcal{M}$ has low-dimensional structure.

However, identifiability alone does not guarantee stable super-resolution. Additional structural conditions on $\mathcal{M}$ are required (see Section \ref{subsec: intro results}). The following simple example also illustrates this point.

\begin{eg}
    Define \[\mathcal{M}_1 = \{\psi\in \mathcal{S}': \ g_k(\psi)=0,\ \forall |k|>K_L  \}.\]
    Then Condition \ref{cond: injective} holds. For $\psi_1\in\mathcal{M}_1$ and $K_H>K_L$, the resolution-enhancing map is given by 
    \begin{align}
        \mathcal{L}\bra{G_L(\psi)} = (0,\cdots,0,g_{-K_L}(\psi),\cdots,g_{K_L}(\psi),0,\cdots,0).
    \end{align}
    In this case, super-resolution is impossible since no high-frequency information can be obtained from the low-frequency part.
\end{eg}

\subsection{Introduction to main results}\label{subsec: intro results}
In this paper, we study the stability of super-resolution in the presence of noise. We propose a model-based super-resolution framework (Model-SR; see Figure 2) to solve the super-resolution problem and to analyze its stability. The key idea is to exploit the low-dimensional latent structure of the target signals. Specifically, we introduce the following definition.

\begin{defn}
We say that the signal space $\mathcal{M}$ is modelable if there exist a compact parameter space $\Theta\subset\mathbb{R}^m$ and a map $\mathcal{P}:\Theta\to \mathcal{M}$ such that $\mathcal{P}$ is surjective and continuous. We call $\mathcal{P}$ the model map, and $(\Theta,\mathcal{P})$ a modeling pair. We define the low-resolution map as $\mathcal{P}_L = G_L\circ\mathcal{P}$ and the high-resolution map as $\mathcal{P}_H = G_H\circ\mathcal{P}$.
\end{defn}

\begin{figure}[ht]
    \centering    
    \begin{tikzcd}[row sep=small,column sep = large]
              &        &\HL \arrow[dd,shift left=1.5,red,"\mathcal{L}"] \arrow[from=dd,"\mathcal{Q}"]\\ 
        \Theta\arrow[urr,bend left=15, "\PL"] 
        \arrow[drr, bend right=15,"\PH",swap]\arrow[r,"\mathcal{P}"]&
        \mathcal{M}\arrow[ur,"G_L"]\arrow[dr,"G_H",swap]&         \\ 
                    &                &\HH 
    \end{tikzcd}
    \caption{Model-based super-resolution framework.}
    \label{fig: model-based SR}
\end{figure}
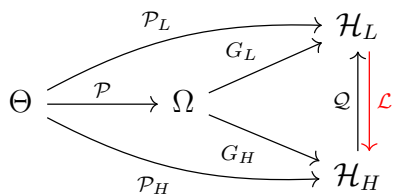

This definition formalizes the assumption that the signal space admits a low-dimensional structure parameterized by $\Theta$. One may extend this viewpoint by assuming that $\mathcal{M}$ is a low-dimensional manifold embedded in an infinite dimensional space. Here, we suppress this generality to focus on the essential difficulties of the super-resolution problem. Moreover, for reconstructing a specific signal, a local chart typically suffices.

Given a modeling pair $(\Theta,\mathcal{P})$ for $\mathcal{M}$, and a noisy measurement $y = \PL(\theta^*)+W$, where the noise $W$ satisfies $\|W\|<\sigma$. The super-resolution task is to reconstruct the high-frequency data $\PH(\theta^*)$.  
We assume that the noise level $\sigma>0$ is known in our theoretical analysis. Beyond this bound, we make no assumptions on $W$. Although leveraging prior information about the noise for denoising is practical in many super-resolution tasks, it is not the focus of this paper. 

Under our model-based super-resolution framework, the problem is solved in the following two steps: 

\begin{enumerate}
  \item [Step 1]\textbf{Parameter estimation}.
We define the set of $(\Theta,\sigma)$-admissible parameters for a measurement $y$ (cf.~Definition~\ref{def-addmi}) as all parameters consistent with $y$ under noise level $\sigma$. This admissible set can be written informaly as $\mathcal{P}_{L, \sigma}^{-1} (y)$, where $\mathcal{P}_{L, \sigma}^{-1}$ denotes the preimage operator, for the ease of notation. In the absence of additional prior information about the noise, all $(\Theta,\sigma)$-admissible parameters are treated as equally plausible instances of the ground truth $\theta^*$. However, different estimators may exhibit different stability guarantees, depending on specific separation conditions (see Section~\ref{sec: point source}-\ref{sec: beyond pts} for three concrete models). 

  \item [Step 2]\textbf{Resolution enhancement.}
  Given the estimate $\hat{\theta}$ from Step~1, we form the super-resolved signal as $\mathcal{P}_H(\hat{\theta})$.
\end{enumerate}

In summary, under our model-based framework, the super-resolution task is to construct the resolution-enhancing map
\begin{align}\label{eq: L}
\mathcal{L} = \PH \circ \mathcal{P}_{L, \sigma}^{-1}. 
\end{align}

In Theorem \ref{cor: repre}, we derive local Lipschitz stability for the super-resolution problem in Step 1 for a ideal modeling pair where the model map $\mathcal{P}$ is a local bijection. We further derive quantitative bounds for several concrete models, including point sources, finite rate of innovation (FRI) signals, and certain classes of continuous signals. Our results show that separation conditions on the parameters are necessary for stability: specifically, the ground truce parameter ${\theta^*}$ must satisfy suitable separation constraints. 

Following the above results, we conclude that a general model-based super-resolution framework does not, by itself, guarantee stability or robustness. Additional modeling conditions are needed. A classical route is to impose a separation condition to ensure robustness. In practice, however, choosing an appropriate separation threshold, which depends on the noise level and the specific signal model, is challenging and often lacks an explicit form. This condition can instead be enforced effectively via sparsity. More specifically, we propose sparse $\sigma$-compatible modeling pairs; see Definition \ref{def-sparse-model}. Under a compatible modeling pair $(\tilde{\Theta}, \tilde{\mathcal{P}})$, the parameter estimation step (Step 1) is replaced by solving the following sparsity-promoting non-convex optimization problem \begin{equation}\label{opt_problem}
  \min_{\theta \in \widetilde{\Theta}} \|\theta\|_{\ell_0}
  \quad \text{subject to} \quad
  \|\widetilde{\mathcal{P}}_L(\theta) - y\| \le \sigma.
\end{equation}
Sparse $\sigma$-compatible modeling pairs can then be constructed  from solutions (not necessarily unique) to this problem.
In Theorem~\ref{cor: modeling}, we establish local Lipschitz stability for the super-resolution problem for $\sigma$-compatible modeling pairs, which include the sparse-compatible ones as a special case.


We also note that we focus throughout on a finite high-frequency cut-off $K_H$. In the limit $K_H \to \infty$, super-resolution reduces to recovering the full signal—equivalently, the ground-truth parameter $\theta^*$. This is a parameter estimation problem, which forms the first step of our model-based super-resolution framework, and it is typically less stable than reconstructing only the high-frequency content up to a finite $K_H$. In the special case of point sources, even determining the exact number of sources is highly ill-posed. By contrast, if we are only interested in frequencies up to the cut-off $K_H$, errors in the estimated number of sources need not affect the recovered band-limited component (cf. Theorems \ref{cor: modeling}), as long as $K_H$ is not too large.  We refer to \cite{LIU2022402,9410626} for related studies on determining the number of sources in the point-source model.


\subsection{Connection with Existing Literature}

Classical studies of super-resolution focus on parameter estimation for the point source model. Originating from Prony’s method \cite{Prony-1795}, subspace techniques such as MUSIC \cite{1143830}, ESPRIT \cite{32276}, and the Matrix Pencil method \cite{hua1990matrix} were developed for high-resolution reconstruction. Their analysis under noise is intricate; see \cite{li2020super,liao2016music,Moitra:2015:SEF:2746539.2746561}. Recent years have also witnessed new variants and extensions of these classical approaches; see, for example, \cite{katz2023decimated,fei2023iff,fei2025scan}. Within our framework, Model-SR reduces to the point source setting when $\mathcal{M}$ consists of point sources, and extends naturally to signals of finite rate of innovation \cite{1003065,doi:10.1137/15M1042280}.

Theoretical advances for the point source model are substantial. Results in \cite{donoho1992superresolution,demanet2015recoverability,batenkov2019super,batenkov2020conditioning,li2021stable} characterize reconstruction stability from a minimax optimality perspective. Recent work introduces the computational resolution limit \cite{9410626,LIU2022402,Liu_2021}, providing quantitative criteria for phase transitions between success and failure under noise; see also \cite{liu2023super,LIU2024101673}. Section \ref{sec: point source} discusses its relation to Model-SR.

In recent years, \emph{compressive sensing} achieved substantial success across a range of practical applications. Its fundamental principle rests on the observation that a measured signal admits a sparse representation with respect to a suitable basis in the signal space \cite{candes2008restricted,candes2008introduction}. The development of compressive sensing subsequently inspired a variety of sparsity-based methodologies for super-resolution, including LASSO, total variation (TV), atomic norm minimization, and B-LASSO \cite{10.2307/2346178,candes2014towards,candes2013super,tang2014near,chi2020harnessing,de2012exact,YANG2018509,doi:10.1137/17M1147822}. Under appropriate minimum-separation conditions, these approaches guarantee exact recovery and stability. They are typically formulated as  
\begin{align}\label{eq: traditional opt.}
    \min_{x\in\mathcal{M}} \|x\|_s 
    \quad \text{s.t. } 
    \|G_Lx - y\|_* \le \gamma,
\end{align}
where \( y \in \mathcal{H}_L \), \(\|\cdot\|_*\) denotes a chosen data-fidelity norm, \(\|\cdot\|_s\) a sparsity-promoting regularizer, and \(\gamma\) a prescribed tolerance.  

In the \emph{Model-SR} framework, introducing a low-dimensional parameterization effectively reduces the dimensionality of the optimization problem, thereby inducing implicit sparsity-promoting regularization. While convex relaxations of \eqref{eq: traditional opt.} ensure stable recovery under strong separation assumptions, Model-SR relaxes these constraints at the expense of solving a generally nonconvex optimization problem. Nevertheless, the sparsity-promoting Model-SR in (\ref{opt_problem}) can be viewed as a form of compressive sensing for nonlinear measurements.


In computer vision, single-image super-resolution (SISR) has been extensively studied. Example-based methods \cite{5459271,freeman2002example,1315043,6751349,yang2012coupled} exploit low-dimensional manifolds underlying images of different resolutions. Classical SISR also includes prediction-based \cite{keys1981cubic,LanczosFilteringinOneandTwoDimensions}, statistical \cite{5430911}, and sparse representation methods \cite{4587647,yang2010image}. Despite this progress, quantitative stability guarantees for super-resolution remain limited. 
In our setting, low- and high-resolution spaces are images of a common signal space under sampling operators, connected by the resolution-enhancing map $\mathcal{L}$ (cf. \eqref{eq: L}). For the concrete examples we study, we derive explicit quantitative stability estimates.


Deep learning has transformed SISR in recent years. SRCNN \cite{7115171} introduced CNN-based approaches mapping low- to high-resolution images. Subsequent work developed deeper networks \cite{kim2016accurate,tai2017image}, U-Net variants \cite{hu2019runet}, adaptive architectures \cite{8306452,dahl2017pixel,9578711}, GAN-based methods \cite{ledig2017photo}, and sparsity-driven models \cite{wang2015deep,7532596,gu2015convolutional}. Surveys include \cite{wang2020deep,yang2019deep}. Connections with Model-SR will be elaborated in Section \ref{sec: extension}. Despite their strong empirical performance, these approaches still lack a rigorous theoretical foundation. 


\subsection{Organization of the paper}
In Section \ref{sec: model-sr}, we introduce the model-based super-resolution framework and the mathematical theory for the proposed framework, with a focus on the stability estimate. We investigate the point source model within Model-SR in Section \ref{sec: point source}. We extend the discussion to signals with a finite rate of innovation and signals having a specific continuous form in the physics domain in Section \ref{sec: beyond pts}. We conduct numerical experiments in Section \ref{sec: Numerical experiments}. In Section \ref{sec: extension}, we discuss several extensions of the Model-SR. The paper concludes with a discussion of the proposed framework in Section \ref{sec: discussion}.

\subsection{Notations}
Throughout the paper, we denote $\|\cdot\|$ the $\ell_2$ norm and $\|\cdot\|_{op}$ the operator norm. We denote $\delta_\theta$ for Dirac measure with support at $\{\theta\}$. For an operator $\mathcal{A}$, we denote $D\mathcal{A}$ the Fr\'echet derivative of $\mathcal{A}$. For a set $U$, $\mathcal{A}|_U$ represents the restriction of $\mathcal{A}$ on $U$. We use the notation $C^k(U,V)$ for $k$-times continuously differentiable functions defined from $U$ to $V$. For matrices $A,B \in \comp^{n\times n}$, $A\preccurlyeq B$ means $B-A$ is positive semi-definite. We use $\sigma_{\min}(A)$ to denote the smallest singular value of $A$. We denote the identity matrix as $\mathcal{I}$ and the identity map as $id$. The notation $B(a,r)$ represents the closed ball centered at $a$ with radius $r$. The notation $m\gtrsim n $ means that there exists a constant $C>0$, such that $m\ge C\cdot n$. We denote the Fourier transform of a function $f(x)$ as $\mathcal{F}[f](\omega)$, defined by $\mathcal{F}[f](\omega) = \int_{\Real} f(x)e^{-2\pi i \omega x}dx$. Finally, we denote $[-\frac{1}{2},\frac{1}{2}]_{*}$ the closed interval $[-\frac{1}{2},\frac{1}{2}]$ equipped with the wrap-around distance 
$d_{\mathbb T}(a,b) = \min_{M\in\mathbb Z}|a-b-M|.$

\section{Model-based Super-resolution Framework}\label{sec: model-sr}
In this section, we develop the mathematical theory of the model-based super-resolution framework (Model-SR). Our presentation is restricted to one dimension for ease of presentation. The generalization to higher dimensions is straightforward.

\subsection{Mathematical Model for Model-SR} \label{subsec: framework}

Recall the definition of a modelable signal space in Section~\ref{subsec: intro results}. We view the signal space $\mathcal{M}$ as a finite-dimensional manifold embedded in the infinite-dimensional space $\mathcal{S}'(\mathbb{R})$. To characterize the dimension of $\mathcal{M}$, we introduce the following definition.

\begin{defn}\label{def: dim M}
    For modelable signal space $\mathcal{M}$, we say that $\mathcal{M}$ has intrinsic dimension $m$ if there exists a modeling pair $(\Theta,\mathcal{P})$ with parameter space $\Theta\subset\Real^m$ and $\operatorname{dim} \Theta = m$ satisfying that for any $\psi\in\mathcal{M}$, there exists a open neighborhood $U_\psi$ and a discrete space $\Lambda_\psi$ such that $\mathcal{P}^{-1}(U_\psi)=\bigsqcup_{d\in \Lambda_\psi} V_d$ and $\mathcal{P}|_{V_d}$ is a bijection for every $d\in \Lambda_\psi$.\footnote{The symbol $\bigsqcup$ denotes the disjoint union.} 
    Further, we refer to such a modeling pair $(\Theta,\mathcal{P})$ as the ideal modeling pair.
    
\end{defn}
\begin{rmk}
    By the continuity of $\mathcal{P}$ and the compactness of $\Theta$, it is easy to show $\mathcal{P}|_{V_d}$ is a homeomorphism for every $d\in \Lambda_\psi$.
\end{rmk}

Notice that in the definition given above, we define the intrinsic dimension through the local bijection (homeomorphism) instead of the global one. This is because of the symmetry property of the model map $\mathcal{P}$, as in the following example.
\begin{eg}
    Let $\Theta=[0,1]^2$, $\theta^{(1)}=(\theta^*_1,\theta^*_2)$, $\theta^{(2)}=(\theta^*_2,\theta^*_1)$, where $\theta^*_1\ne\theta^*_2$. Define $\mathcal{P}(\theta) = \delta_{\theta_1}+\delta_{\theta_2}$. Then, $\mathcal{P}(\theta^{(1)})=\mathcal{P}(\theta^{(2)})$, though $\theta^{(1)}\ne\theta^{(2)}$.
\end{eg}

Now, consider the noisy low-resolution measurement $y\in\comp^{2K_L+1}$ given by
\begin{align}\label{noisy measurement}
    y = G_L (\psi) + W = \PL(\theta^*)+W,
\end{align}
where $\theta^*\in \Theta$, and $W = \bra{W_{-K_L},\cdots,W_{K_L}}$ is the noise vector with $\|W\|<\sigma$. We refer $\theta^*$ as the ground-truth parameter. Note that in practice, the ground-truth parameter may not be uniquely solvable from noisy measurement. 
To relate the noisy signal and the parameter to recover, we introduce the following concept.
\begin{defn} \label{def-addmi}
    Given a low-resolution noisy measurement $y$, for any chosen modeling pair $(\Theta,\mathcal{P})$, we say $\theta\in\Theta$ is a $(\Theta,\sigma)$-admissible if 
    \begin{align}
        \|\PL(\theta)-y \| \le \sigma.
    \end{align}
\end{defn}

It is clear that $(\Theta,\sigma)$-admissible parameters can be solved via the following nonlinear least-squares problem:
\begin{align}
    \min_{\theta \in \Theta}~ \frac{1}{2}\|\mathcal{P}_L(\theta) - y\|^2.
\end{align}



\subsection{Stability Estimate for an Ideal Modeling Pair} \label{subsec: general stability}
Recall that for an ideal modeling pair, the model map $\mathcal{P}: \Theta \to \mathcal{M}$ is a ``local" bijection. We first analyze the local property of the low-resolution map and show the Lipschitz continuity of its inverse. 
Denote 
\[
\|D\PH\|_{op}= \max_{\theta \in U, \|z\| \leq 1}\left\|D\PH(\theta) z\right\|.   
\]

\begin{prop}\label{thm: abstract stability}
    Assume that $U \subset \Real^m$ is a convex compact set. Consider $\PL \in C^1(\Real^m,\HL)$ satisfying that 
    \begin{itemize}
        \item $\PL|_U$ is injective,
        \item $D\PL(\theta)$ is injective for all $\theta\in U$.
    \end{itemize}
   
    Then, for every $\theta,\theta'\in U$, there exists $C_{U}>0$ such that
    \begin{align} \label{stab: theta}
        \|\theta-\theta'\| \le C_{U} \cdot \|\PL(\theta)-\PL(\theta')\|.
    \end{align}
    Further, we have 
    \begin{align}
        \|\PH(\theta)-\PH(\theta') \| \le C_{U}\cdot \|D\PH\|_{op}\cdot \|\PL(\theta)-\PL(\theta')\|.
    \end{align}
\end{prop}

The above proposition is a consequence of Theorem 2.1 (cf. \cite{alberti2022inverse}). For the sake of completeness and the convenience of readers, we offer the proof in Appendix \ref{prf: general stability}. We note that condition (\ref{stab: theta}) plays an analogous role to the Restricted Isometry Property (RIP) in ensuring stable recovery of sparse signals from linear measurements (cf. \cite{candes2008restricted,foucart2013invitation}). A uniformly bounded constant $C_U$ ensures that a parameter $\theta$ can be stably reconstructed from its low-frequency measurement $\mathcal{P}_L(\theta)$.

As a consequence of Proposition \ref{thm: abstract stability}, we have the following Lipschitz stability estimate for $(\Theta,\sigma)$-admissible parameters for the ideal modeling pair.

\begin{thm}\label{cor: repre}
    Assume that $U \subset \Real^m$ is a convex compact set. Consider $\PL\in C^1(\Real^m,\HL)$ satisfying that 
    \begin{itemize}
        \item $\PL|_U$ is injective,
        \item $D\PL(\theta)$ is injective for all $\theta\in U$.
    \end{itemize}
     Let $\hat{\theta}\in U$ be a $(\Theta, \sigma)$-admissible parameter for the noisy measurement (\ref{noisy measurement}), then 
    \begin{align}\label{stability est repre}
        \| \PH(\hat{\theta})-\PH(\theta^*) \| < 2C_{U}\cdot \|D\PH\|_{op}\cdot\sigma.
    \end{align}
\end{thm}

\begin{proof}
    Notice that 
    \begin{align}
        \|\PL(\hat{\theta})-\PL(\theta^*) \| \le \|\PL(\hat{\theta})-y\| +\|\PL(\theta^*)-y\| <2\sigma.
    \end{align}
    By Theorem \ref{thm: abstract stability}, we have 
        \begin{align}
        \|\PH(\hat{\theta})-\PH(\theta^*) \| \le C_{U}\cdot \|D\PH\|_{op}\cdot \|\PL(\hat{\theta})-\PL(\theta^*) \| < 2C_{U}\cdot \|D\PH\|_{op} \cdot \sigma.
    \end{align}
\end{proof}
We note that the Lipschitz constant of the resolution-enhancing map ($\PH\circ\mathcal{P}_{L, \sigma}^{-1}$) can be naturally decomposed into two parts: $C_U$, reflecting the stability of signal-space modeling from low-resolution samples, and $\|D\PH\|_{op}$, governing extrapolation stability in the frequency domain. The former depends on $K_L$, the latter on $K_H$. This decomposition makes explicit  how stability scales with the super‑resolution factor (SRF). The inverse step of reconstructing parameters from low-resolution measurements can be ill‑posed, depending on the parameters to be recovered, and thus requires additional restrictions. For point‑source models, a natural restriction is a minimum‑separation condition on source locations. Super‑resolution theory characterizes the separation thresholds that ensure stable parameter recovery.

\subsection{Stability for Compatible Modeling Pairs}\label{sec: sta_spase}
As discussed in the previous section, the parameter estimation for an ideal modeling pair may still be unstable. Such instability typically arises when the parameters lack a certain \emph{separation property}. For instance, accurately recovering the locations of two point sources whose separation distance falls below the Rayleigh length becomes highly challenging under realistic noise levels.

Here, we interpret the notion of parameters having a \emph{good separation property} as meaning that the parameters can be stably reconstructed from the low-resolution measurements at a given noise level. This concept is closely related to the \emph{computational resolution limit} developed for the point source model, which quantifies the gap between signals generated by $n$ versus $n-1$ sources under noise via a minimum separation condition (see, e.g., \cite{LIU2022402}).

Parameter estimation with poor separation using an ideal modeling pair can lead to instability, making super-resolution unreliable. However, this issue can be mitigated by using a compatible modeling pair with fewer parameters. The underlying idea is that reducing the number of parameters increases their effective separation, thereby improving stability.

\begin{defn} \label{def-sparse-model}
    We call that a modeling pair $(\widetilde{\Theta},\widetilde{\mathcal{P}})$ is $\sigma$-compatible to the the measurement $y$  if $\exists~\tilde{\theta}\in \widetilde{\Theta}$ such that 
    \begin{align}
     \|\widetilde{\mathcal{P}}_L(\tilde{\theta})-y \|\le\sigma.
    \end{align}
    Further, we call $(\widetilde{\Theta}^{(s)},\widetilde{\mathcal{P}}^{(s)})$ a sparse $\sigma$-compatible modeling pair if
    \begin{align*}
        \operatorname{dim}\left(\widetilde{\Theta}^{(s)}\right) = \min_{\widetilde{\Theta}\in \Lambda} \operatorname{dim}\left(\widetilde{\Theta}\right) ,~\Lambda = \left\{ \widetilde{\Theta}: \text{$(\widetilde{\Theta},\widetilde{\mathcal{P}})$ is $\sigma$-compatible to $y$} \right\}
    \end{align*}
\end{defn}
Note that a $\sigma$-compatible modeling pair to a measurement $y$ need not contain the ground truth parameter $\theta^*$ that generates $y$, yet it can still approximate $y$ well. 

Starting from an \emph{over-parameterized} model, that is, when $\dim(\widetilde{\Theta})$ exceeds (or equals) the intrinsic dimension of the signal space, a sparse $\sigma$-compatible modeling pair can be obtained by solving the folllowing optimization problem
\begin{align*}
    \min_{\theta \in \widetilde{\Theta}} \|\theta\|_{\ell_0}
    \quad \text{subject to} \quad
    \|\widetilde{\mathcal{P}}_L(\theta) - y\| \le \sigma,
\end{align*}
and then restricting $\widetilde{\Theta}$ to a minimal subspace that contains a solution to this problem. 
Intuitively, the $\ell_0$-regularization term acts as a \emph{model selector}, favoring the simplest model, namely, the one with the fewest parameters, the sparse $\sigma$-compatible modeling pair. Note that such pairs need not be unique.

We now address the stability of super-resolution for compatible modeling pairs, which includes the sparse compatible case as a special case. The following corollary extends the preceding stability estimate to any compatible modeling pair $(\widetilde{\Theta}, \widetilde{\mathcal{P}})$.

\begin{thm}\label{cor: modeling}
    Let $(\widetilde{\Theta},\widetilde{\mathcal{P}})$ be a $\sigma$-compatible modeling pair to $(\Theta,\mathcal{P})$, and let $\tilde{\theta}\in \widetilde{\Theta}$ satisfy
    \[
    \|\widetilde{\mathcal{P}}_L(\tilde{\theta})-y \|<\sigma.
    \]
    Assume that $\widetilde{U}$ is a convex compact set and that 
    \begin{itemize}
    \item $\widetilde{\mathcal{P}}_L|_{\widetilde{U}}$ is injective.
    \item $D\widetilde{\mathcal{P}}_L|_{\widetilde{U}}$ is injective.
    \item There exists $C_{\tilde{U}}>0$ such that
    \begin{align}
        \|\theta-\theta'\| \le C_{\tilde{U}} \cdot \|\widetilde{\PL}(\theta)-\widetilde{\PL}(\theta')\|.
    \end{align}
    \end{itemize}
 Assume that $\hat{\theta}$ satisfies $\|\widetilde{\mathcal{P}}_L(\hat{\theta})-y \|<\sigma$. Then 
    \begin{align} \label{eq-ph8}
        \| \widetilde{\mathcal{P}}_H(\hat{\theta})-\widetilde{\mathcal{P}}_H(\tilde{\theta}) \| < 2C_{\widetilde U}\cdot \|D\widetilde{P}_H\|_{op}\cdot\sigma.
    \end{align}
\end{thm}

\begin{proof}
    By Theorem \ref{cor: repre}, we have
    \[
    \| \widetilde{\mathcal{P}}_H(\hat{\theta})-\widetilde{\mathcal{P}}_H(\tilde{\theta}) \| \le C_{\widetilde U}\cdot \|D\widetilde{P}_H\|_{op}\cdot \|\widetilde{\mathcal{P}}_L(\hat{\theta})-\widetilde{\mathcal{P}}_L(\tilde{\theta})\| 
    < 2C_{\widetilde U}\cdot \|D\widetilde{P}_H\|_{op}\cdot\sigma..
    \]
\end{proof}

In the theorem above, the Lipschitz constant $C_{\tilde{U}}$ is model-dependent. In particular, compared with the ideal modeling pair, employing a \emph{sparse} $\sigma$-compatible modeling pair can lead to a significant improvement in parameter separability. Consequently, the Lipschitz constant $C_{\tilde{U}}$ may be substantially smaller. Likewise, for the extrapolation component, the reduction in the number of parameters can also contribute to a decrease in the corresponding Lipschitz constant (cf.\ the three concrete models discussed in the subsequent sections). This theoretical characterization underscores the critical role of \emph{sparse model selection} in achieving enhanced stability in the super-resolution procedure.
The following numerical example illustrates this idea more concretely.

\begin{eg}\label{example: 3 sources using 2}
Consider a signal consisting of three point sources,
\[
\delta_{0.1} + \delta_{0.5} + \delta_{0.5 + \Delta},
\]
with $\Delta = 1/40$, corresponding to the Rayleigh length for a cutoff frequency $K_{\mathrm{cut}} = 20$. We observe its low-frequency measurements with $K_L = 5$, and we adopt a modeling pair corresponding to two well-separated point sources. We then set the high-frequency cutoff at $K_H = 30$. The reconstruction results are shown in Figure~\ref{fig: two point approx}. We observe that the two-source approximation captures the main features of the original signal, and the subsequent super-resolution step yields a reasonable high-resolution reconstruction.  As $K_H$ increases beyond a certain threshold, however, the mismatch between the assumed and true modeling pairs becomes apparent, and the super-resolution reconstruction error grows significantly, reflecting the intrinsic ill-posedness of recovering point sources below the Rayleigh length.
\end{eg}

\begin{figure}[!ht]
    \centering
    \subfloat[Real part of the approximation\label{fig: real_part}]{
        \includegraphics[width=0.48\linewidth]{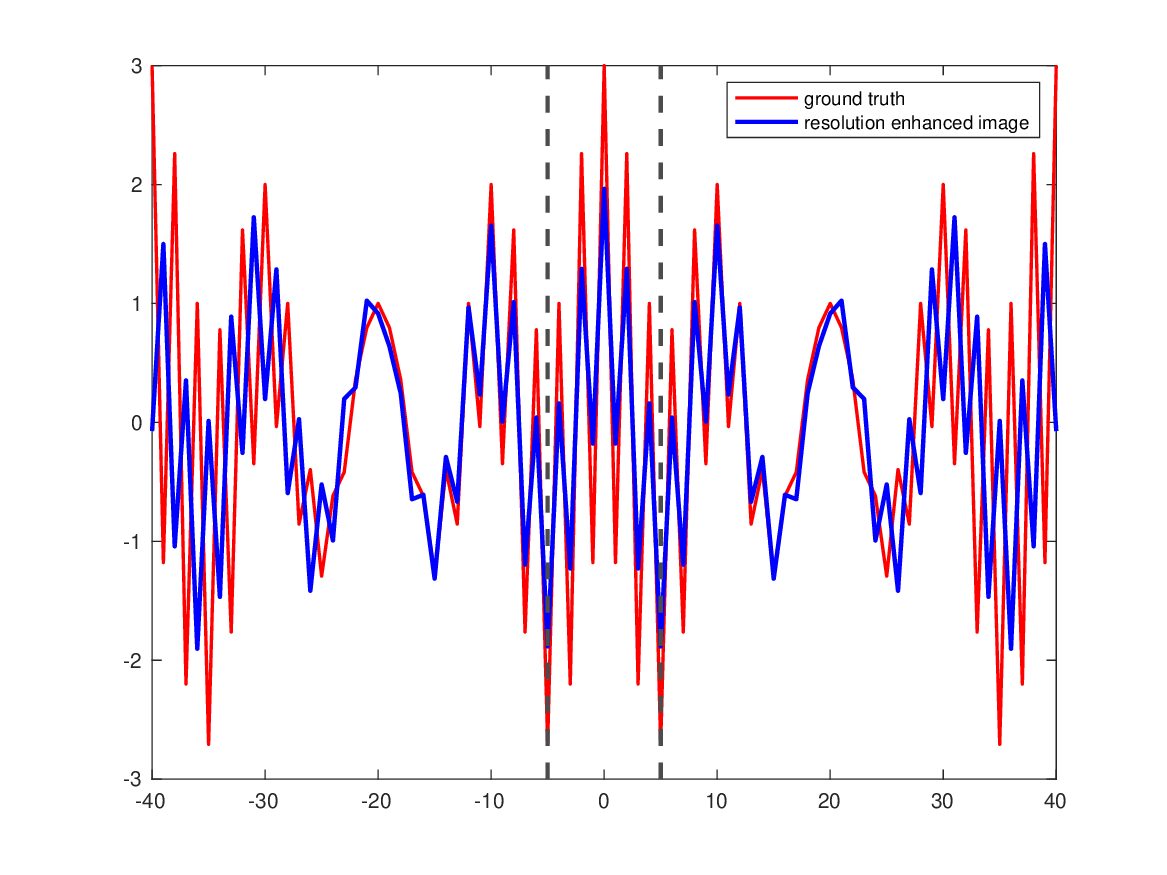}
    }
    \hfill
    \subfloat[Imaginary part of the approximation\label{fig: imag_part}]{
        \includegraphics[width=0.48\linewidth]{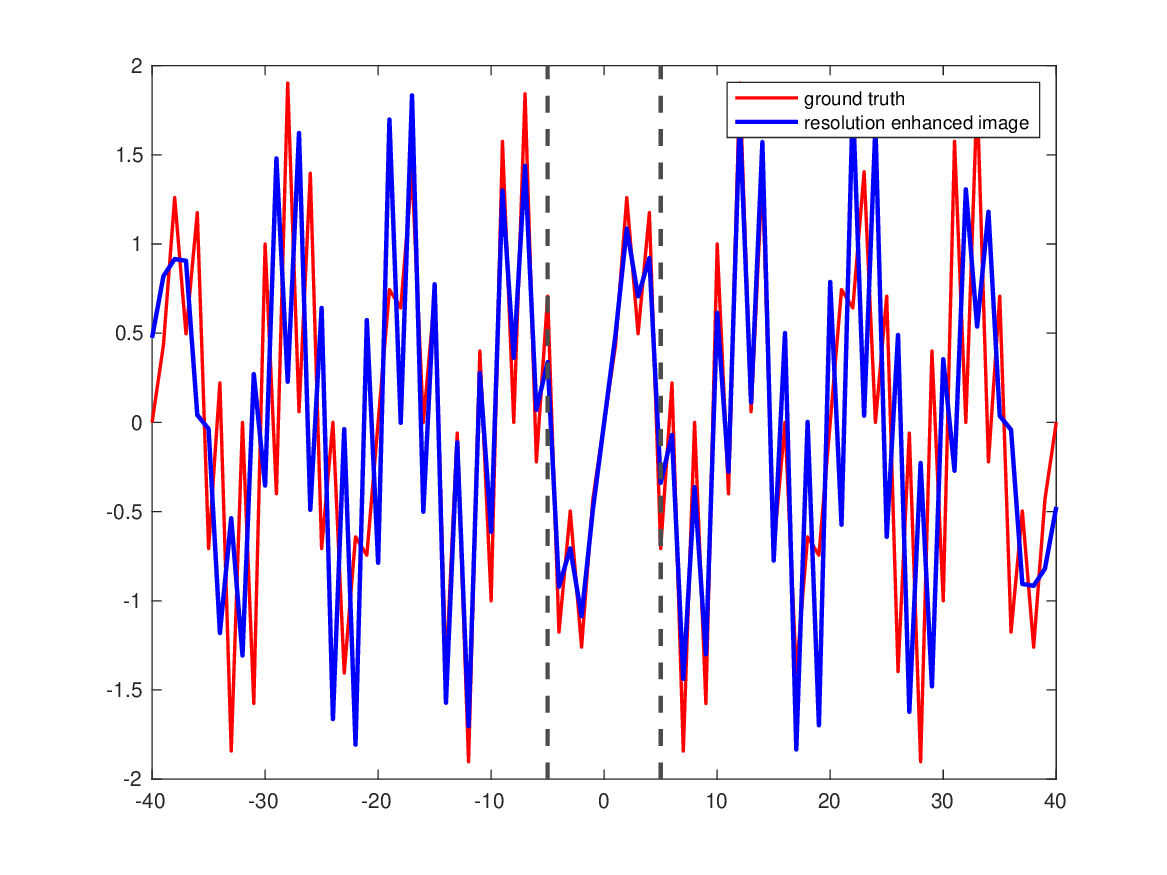}
    }
    \caption{Two-point source approximation of the signal, followed by resolution enhancement up to $\mathrm{SRF} = 6$. The experiment is conducted under $\mathrm{SNR} = 12.76$ (see the definition of SNR in (\ref{eq: snr})).}
    \label{fig: two point approx}
\end{figure}

\begin{rmk}
In the stability estimate \eqref{eq-ph8}, we use a $\sigma$-admissible parameter $\tilde{\theta}$ from the compatible modeling pair as the reference “ground truth”, in contrast to \eqref{stability est repre}, which uses the actual ground-truth parameter $\theta^*$. 
\end{rmk}
Within the framework of sparse compatible modeling pairs, when the ground truth $\theta^*$ is well separated, one can still select a sparse and well-separated estimate, and super-resolution is stable for such signals. When $\theta^*$ is not well separated, the situation is subtler. For moderate or large noise, the sparse $\sigma$-admissible set may contain well-separated parameters. Then the optimization problem (\ref{opt_problem}) can be stably solved and this leads to stable super-resolution. 
However, as the noise level $\sigma \to 0$, the $(\Theta,\sigma)$-admissible set shrinks and concentrates around $\theta^*$, so no admissible parameter can be well separated. As a consequence the optimization problem (\ref{opt_problem}) remains ill-conditioned. This leads to instability of super-resolution,  consistent with the intuition that achieving super-resolution in the sub-Rayleigh regime is inherently challenging, although theorectical possible.

\subsection{The Optimization Problem} \label{subsec: optimization}

In this section, we characterize the landscape of the optimization problem (\ref{opt_problem}). We restrict attention to the case where $(\Theta,\mathcal{P})$ is a sparse $\sigma$-compatible modeling pair, and focus on the local property of $\varphi(\theta)$. A treatment of general $\sigma$-compatible modeling pairs and a global landscape analysis is beyond the scope of this paper. We also note that there exist numerous numerical optimization methods for solving the nonlinear least-squares problem (\ref{opt_problem}); see, for example, \cite{mohammad2019brief} for a brief survey.

Let the noisy low-resolution measurement be given in (\ref{noisy measurement}). We first write the objective function as 
\begin{align}
    \varphi(\theta) = \frac{1}{2}\cdot \| \PL(\theta)-y \|^2 = \frac{1}{2} \cdot \sum_{k=-K_L}^{K_L} \left|\mathcal{P}_{L,k} (\theta)-y_k\right|^2,
\end{align}
where $y_k = \mathcal{P}_{L,k}(\theta^*)+W_k$. The following theorem shows that for any solution to (\ref{opt_problem}), under a certain noise level, the objective function is locally $\nu_l$-convex and $\nu_u$-smooth in the neighborhood of the solution.
\begin{thm}{\label{thm: abs_convergence}}
    Assume that $U\subset\Real^m$ is a compact set.  Consider $\PL \in C^2(\Real^m,\HL)$ satisfying that 
    \begin{itemize}
        \item $\PL|_{\Theta}$ is injective,
        \item $D\PL(\theta)$ is injective for all $\theta\in U$.
    \end{itemize}
    Let $\hat{\theta}\in U$ be a solution to optimization problem (\ref{opt_problem}), then 
    \begin{itemize}
        \item $\hat{\theta}$ is an admissible solution.
        \item there exists a neighborhood of $\hat{\theta}$, say $U_{\hat{\theta}}\subset U \subset\Theta$, and $\nu_u,\nu_l >0$ such that 
    \begin{align}
        \nu_l \mathcal{I} \preccurlyeq \nabla^2\varphi(\theta) \preccurlyeq \nu_u \mathcal{I}, \quad \forall \theta\in U_{\hat \theta},
    \end{align}
    provided $\sigma \le \frac{\sigma^2_{\min}\bra{D\PL(\hat{\theta})}}{\|\xi\|}$, where $\xi = (\|\nabla^2\mathcal{P}_{L,-K_L}\|_{op},\cdots,\|\nabla^2\mathcal{P}_{L,K_L}\|_{op})$. 
    \end{itemize}
    
\end{thm}

\begin{rmk}
    The $(\Theta,\sigma)$-admissibility of $\hat{\theta}$ implies the stability estimate (\ref{stability est repre}) holds provided that $D\PH$ is bounded. If we consider any compatible modeling pair $(\widetilde{\Theta},\widetilde{\mathcal{P}})$, the analog of admissibility and stability estimate still holds.
\end{rmk}
Following the standard convergence analysis, the theoretical convergence rate of different optimization algorithms can be derived in this case. For instance, for suitable initialization and step size, gradient descent method has convergence rate $\Ord\bra{\bra{1-\frac{\nu_l}{\nu_u}}^t}$, and Nesterov accelerated gradient descent has convergence rate $\Ord\bra{\bra{1-\sqrt{\frac{\nu_l}{\nu_u}}}^t}$.

\begin{rmk}
 A limitation of our model-based super-resolution framework is the need to solve a nonconvex optimization problem, in contrast to convex approaches \cite{10.2307/2346178,candes2014towards,candes2013super,tang2014near,chi2020harnessing,de2012exact,YANG2018509,doi:10.1137/17M1147822}. However, we argue that, due to the inherent nonconvexity of super-resolution, convex relaxations may fail in challenging regimes—such as recovering point sources separated below the Rayleigh limit. With advances in iterative solvers and initialization strategies, model-based nonconvex methods may be better suited for practical applications. 
In addition, non-iterative subspace methods such as MUSIC, ESPRIT, and matrix pencil techniques (cf. \cite{schmidt1986multiple,roy1989esprit,hua1990matrix,stoica2005spectral,candes2014towards}) have demonstrated excellent super-resolution capabilities in low-noise regimes. However, these methods depend crucially on specific signal classes, particularly the point-source (line spectral) model. Moreover, due to heavy computational burdens and unfavorable sample-complexity in high dimensions, their practical use is largely restricted to one- and, to a limited extent, two-dimensional settings.
\end{rmk}

\section{Point Source Model}{\label{sec: point source}}
In this section, we consider the super-resolution problem for the point source model.\\

For the signals generated by $n$ different sources in the interval $[-\frac{1}{2},\frac{1}{2}]_*$ with amplitudes taking values in a closed interval $I\subset [-A_I,A_I]$. Let $\Theta = I^n \times [-\frac{1}{2},\frac{1}{2}]_*^n$ be the parameter space and let $\theta = (\theta_{1,1},\cdots,\theta_{n,1},\theta_{1,2},\cdots,\theta_{n,2})\in \Theta$. We define the model map as 
\begin{align}\label{eq: pts P}
    \mathcal{P}(\theta)= \psi(x) = \sum_{j=1}^n \theta_{j,1}\delta_{\theta_{j,2}}.
\end{align}
Here, $\theta_{j,2}$ represents the position of the point sources, $\theta_{j,1}$ the corresponding amplitude. The intrinsic dimension of the signal space is $2n$. Then, the signal space has an explicit form:
\begin{align}\label{model: point}
    \mathcal{M}  = \left\{ \sum_{j=1}^n \theta_{j,1}\delta_{\theta_{j,2}}: \theta_{j,1}\in I, \theta_{j,2}\in [-\frac{1}{2},\frac{1}{2}]_*  \right\}.
\end{align}
For the grid defined by $\omega_k = k$, we have 
\begin{align}
    g_k = \sum_{j=1}^n \theta_{j,1}e^{-2\pi i \theta_{j,2}k}.
\end{align}
The low- and high-resolution sampling operators, $G_L$ and $G_H$, can be defined by (\ref{def: GL}) and (\ref{def: GH}) respectively. 
Consequently, the noisy low-resolution measurement can be expressed as 
\begin{align}\label{eq: point signal profile}
    y_k = g_k+W_k = \sum_{j=1}^n \theta_{j,1}e^{-2\pi i \theta_{j,2}k}+W_k, \quad k=-K_L,\cdots,K_L,
\end{align}
where $W_k$ is the noise term with $|W_k|<\sigma$.  We assume that $K_L\ge n$. The Rayleigh length of this system is defined as $RL = \frac{1}{2K_L}$. We define the minimum separation distance $d_{\min}$, and the minimum amplitude $m_{\min}$ as 
\begin{align}
    d_{\min} = \min_{j\ne j'}|\theta_{j,2}-\theta_{j',2}|, \quad m_{\min} = \min_j |\theta_{j,1}|.
\end{align}


We notice that in \cite{9410626,LIU2022402}, the authors show that the computational resolution limit for the point source model is given by 
\begin{align}
    \mathcal{D}_{supp} \sim \Ord\bra{\frac{1}{2K_L}\bra{\frac{\sigma}{m_{\min}}}^\frac{1}{2n-1}}.
\end{align}
In recent work \cite{liu2024mathematical}, the authors improve the characterization of $\mathcal{D}_{supp}$ to 
\begin{align}
   \frac{1}{2K_Le\pi} \bra{\frac{\sigma}{m_{\min}}}^\frac{1}{2n-1}<\mathcal{D}_{supp}\le \frac{2.36e}{2K_L} \bra{\frac{\sigma}{m_{\min}}}^\frac{1}{2n-1}
\end{align}

Thus, we note that when the minimum separation between the $n$ point sources exceeds the computational resolution limit, the ideal modeling pair introduced above coincides with a sparse $\sigma$-compatible modeling pair. A complete characterization of the corresponding stability is provided below.

\begin{thm}{\label{thm: stability point source}}
    Let $\Theta = I^n \times [-\frac{n-1}{4K_L},\frac{n-1}{4K_L}]^n$, and $\mathcal{P}$ be defined in (\ref{eq: pts P}), considering the signal having the form in (\ref{eq: point signal profile}). Assume that the following separation condition is satisfied 
    \begin{align}
        d_{\min} = \min_{j\ne j'}|\theta_{j,2}-\theta_{j',2}|> \frac{2.36e}{2K_L} \bra{\frac{\sigma}{m_{\min}}}^\frac{1}{2n-1}.
    \end{align}
    For any $(\Theta,\sigma)$-admissible parameter $\hat{\theta}$, we have 
    \begin{align} \label{eq-c1}
        \|\hat{\theta}-\theta^*\|\le C_1\cdot\sigma,
    \end{align}
    where 
    $C_1^2 = n\bra{\frac{C_{1,1}(n)}{2K_L\pi \cdot m_{\min}}\bra{\frac{1}{2K_Ld_{\min}}}^{2n-2}}^2+n \bra{C_{1,2}(n)\bra{\frac{1}{2K_Ld_{\min}}}^{2n-1}}^2$ with constants $C_{1,1}(n)$ and $C_{1,2}(n)$ that only depend on the source number $n$. \\
    
    Further, for such $\theta$, we have
    \begin{align}
        \| D\PH\|_{op}^2 \le (2K_H+1)n+\frac{4n\pi^2A_I^2}{3}\cdot\bra{K_H(K_H+1)(2K_H+1)},
    \end{align}
    and
    \begin{align}
        \|D\PH\|_{op} \gtrsim K_H^{3/2}.
    \end{align}
    Let $C_1' = C_1\cdot \|D\PH\|_{op}$, we have
    \begin{align}\label{Lip: point source}
        \|\PH(\theta)-\PH(\theta')\|\le C_1' \cdot \sigma.
    \end{align}
    As a consequence, if $\hat{\theta}$ is a $(\Theta,\sigma)$-admissible parameter, then
    \begin{align}
        \|\PH(\hat{\theta})-\PH(\theta^*)\|<2C_1'\cdot \sigma.
    \end{align}
\end{thm}

\begin{rmk}
   By the theorem above, the Lipschitz constant for the high-frequency data extrapolation of the point source model is of the order $\Ord(K_H^{3/2})$ as $K_H$ grows.
\end{rmk}
\begin{rmk}
    The explicit express of the Lipschitz constant $C_1$ in (\ref{eq-c1}) indicates a phase transition in the stability of parameter recovery from low‑resolution measurements. Specially, for fixed source number $n$, $C_1$ grows polynomially with $2K_Ld_{\min}$ if $2K_Ld_{\min}<1$, and decrease as $2K_Ld_{\min}$ increases beyond one. Here $2K_Ld_{\min}$ is the ratio of the the minimum separation distance and Rayleigh length. 
\end{rmk}

The stability result above characterizes how the stability depends on the minimum separation distance, cutoff frequency. Here, we further point out that the result derived above implies that $\theta$ in the above theorem can be determined exactly from the low-resolution measurement for the noiseless measurement. This implies that the signal $\mathcal{P}(\theta)$ can be exactly recovered. We also notice that the result can be extended to the case when the sources have complex amplitudes. These observations indicate that within the proposed framework, the exact signal recovery does not require the minimum separation distance condition nor the conditions on source signs for the noiseless measurement (this is to be contrasted with the BLASSO strategy, for which a counter-example exists for sources having arbitrary sign and separation distance below $1RL$ \cite{doi:10.1137/17M1147822}). Further, the stability result offers a perspective on how the $\ell_2$ error of a high-resolution signal depends on noise. In the super-resolution literature, various other types of stability results have also been investigated. For instance, the authors in \cite{candes2014towards, candes2013super} derived $\ell_1$-based stability estimates for total-variation-norm-minimization solutions to the super-resolution problem in the point source model, subject to the minimum separation condition.

\section{Going Beyond Point Source Model}\label{sec: beyond pts}
In this section, we discuss the application of the general theory developed in previous chapters on more general models. The modeling pair we pick in this section is assumed to the sparse $\sigma$-compatible modeling pair. 

\subsection{Signals with Finite Rate of Innovation}{\label{sec: finite rate of innovation}}

In this section, we consider the super-resolution problem for signals with a finite rate of innovation (FRI), see e.g. \cite{1003065,doi:10.1137/15M1042280}. We use signals generated by derivatives of Diracs in the physics domain as a typical example for demonstration.\\

We consider the sources in the interval $[0,1)$ with amplitudes taking values in a closed interval $I\subset [-A_I,A_I]$. For the sources corresponding to the $r$-th derivative of delta, $r=0,\cdots,R$, we denote the total number as $n_r$, the source positions as $\{\theta_{r,j,2}\}_{j=1}^{n_r}$, and the amplitudes as $\{\theta_{r,j,1}\}_{j=1}^{n_r}$. We write $N=\sum_{r=0}^R n_r$ for the total number of sources.\\

Let $\Theta = I^N\times[0,1]_*^N$ be the parameter space, we define the model map as 
\begin{align}
    \mathcal{P}(\theta) = \psi(x)= \sum_{r=0}^R \sum_{j=1}^{n_r} \theta_{r,j,1}  \delta^{(r)}_{\theta_{r,j,2}},
\end{align}
where $\theta=(\theta_{1,0,1},\cdots,\theta_{n_R,R,1},\theta_{1,0,2},\cdots,\theta_{n_R,R,2})$, and $\delta^{(r)}$ denotes the $r$-th derivative of $\delta$. Thus, the signal space can be written as $\Omega = \mathcal{P}(\Theta)$. For the grid $\omega_{k} = k$, we have 
\begin{align}
    g_k = \sum_{r=0}^R \sum_{j=1}^{n_r} \theta_{r,j,1}\cdot (-2\pi i k)^r  e^{-2\pi i \theta_{r,j,2}k}.
\end{align}
The low- and high-resolution sampling operators, $G_L$ and $G_H$, are defined by (\ref{def: GL}) and (\ref{def: GH}) respectively. 
Consequently, the noisy low-resolution measurement can be expressed as 
\begin{align}
    y_k = g_k+W_k = \sum_{r=0}^R \sum_{j=1}^{n_r} \theta_{r,j,1}\cdot (-2\pi i k)^r  e^{-2\pi i \theta_{r,j,2}k}+W_k, \quad k=-K_L,\cdots,K_L,
\end{align}
where $W_k$ is the noise term with $|W_k|<\sigma$.  We assume that $K_L\ge N$. The Rayleigh length is defined as $RL=\frac{1}{2K_L}$.\\

Applying Theorem \ref{thm: abstract stability} and Theorem \ref{cor: repre} to this model, we have the following stability estimate.
\begin{thm}{\label{thm: stability diff Dirac}}
    For any given $\theta^* = (\theta_{0,1,1}^*,\cdots,\theta_{R,n_R,1}^*,\theta_{0,1,2}^*,\cdots,\theta_{R,n_R,2}^*)\in \Theta$, let $\Delta_r = \frac{1}{2}\cdot \min_{p\ne q} d_{\mathbb T}(\theta^*_{r,p,2},\theta^*_{r,q,2})$. We assume that $\theta_{r,j,1}^*\ne 0$, for $j=1,\cdots n_r$ and $r=0,\cdots, R$, and $\Delta_r>0$. Let $U = \prod_{r=0}^R \prod_{j=1}^{n_r} \bra{B(\theta_{r,j,1}^*,\frac{|\theta_{r,j,1}^*|}{2})\cap I}\times \prod_{r=0}^R \prod_{j=1}^{n_r} B(\theta_{r,j,2}^*,\Delta_r)$. Then, there exists $C_{U}>0$ such that for any $\theta,\theta'\in U$, 
    \begin{align} \label{eq-cu2}
        \|\theta-\theta'\| \le C_{U} \cdot \|\PL(\theta)-\PL(\theta')\|.
    \end{align}
    In addition, for $\theta\in U$, 
    \begin{align}
        \|D\PH(\theta) \|_{op}^2 \le \sum_{k=-K_L}^{K_L} \sum_{r=0}^R n_r(2\pi k)^{2r}\bra{1+4\pi^2k^2A_I^2},
    \end{align}
    and
    \begin{align}
        \| D\PH(\theta)\|_{op} \gtrsim K_H^{R+3/2}.
    \end{align}
   Furthermore, 
    \begin{align}
        \|\PH(\theta)-\PH(\theta')\|\le C_U\cdot \|D\PH\|_{op} \cdot \|\PL(\theta)-\PL(\theta')\|.
    \end{align}
    As a consequence, if $\hat{\theta}\in U$ is a $(\Theta, \sigma)$-admissible parameter, then 
    \begin{align}
        \| \PH(\hat{\theta})-\PH(\theta^*) \| < 2C_U\cdot \|D\PH\|_{op}\cdot\sigma.
    \end{align}
\end{thm}
\begin{rmk}
The assumption $\theta_{r,j,1}^* \neq 0$ is made for simplicity. It can be relaxed by restricting attention to a reduced signal subspace with fewer active sources.
\end{rmk}

\begin{rmk}
By the above theorem, the Lipschitz constant for the high-frequency data extrapolation of the FRI signals is of the order $\Ord(K_H^{R+3/2})$ as $K_H$ grows. Thus, stability is heavily influenced by the sources associated with the highest order of derivatives of Diracs.
\end{rmk}

The exact characterization of the constant $C_U$ in (\ref{eq-cu2}) is too intricate to present here. Instead, we provide a simple example below to demonstrate its dependence on the separation distance.

\begin{prop} \label{prop4-4}
Consider two sources in the physical domain, $\delta_z^{(r)}+\delta_{z'}^{(r)}$, for $z,z'\in [0,1]$ with $0<|z-z'|<\frac{1}{4K_L}$. We denote $\Delta : =|z-z'| $. Then 
\begin{align} \label{eq-kl1}
    \|\theta -\theta' \| \le \frac{C}{ K_L^{r+5/2}\cdot\Delta}\|\PL(\theta)-\PL(\theta') \|,
\end{align}
for some universal constant $C>0$.

\end{prop}

As a generalization of the point source model, there are few theoretical results for signals with a finite rate of innovation. We notice that the authors consider the on-the-grid setting stability estimate for $R=1$ in \cite{batenkov2023super}. For the general model, the problem is still widely open. 

\subsection{Towards General Signals}{\label{sec: continuous spectrum}}
Previously, we consider the super-resolution problem for signals having discrete forms in the physical domain. In this section, we consider continuous signals in the physical domain.\\

To demonstrate the idea, we consider signals that are probability density functions of Gaussian mixtures with $n$ components. More precisely, let $\Theta = I_1^n\times I_2^n\times[-\frac{1}{2},\frac{1}{2}]_*^n$ be the parameter space. We define the model map as 
\begin{align}
    \mathcal{P}(\theta) = \psi(x) =\sum_{j=1}^n \theta_{j,1}e^{-\frac{(x-\theta_{j,2})^2}{2\alpha^2}},
\end{align}
where $\theta=(\theta_{1,1},\cdots,\theta_{n,1},\theta_{1,2},\cdots,\theta_{n,2})$. Here $\theta_{j,1}$, $\theta_{j,2}$ are the weight and mean of the $j$-th component, and they take values in a closed interval $I_1$ and $[-\frac{1}{2},\frac{1}{2}]_*$, respectively. Additionally, $\alpha>0$ represents the variance and is assumed to be known. The signal space can be written as $\mathcal{M} = \mathcal{P}(\Theta)$. For the grid $\omega_{k} = k$, we have 
\begin{align}
    g_k = \sqrt{2\pi\alpha^2}\cdot\sum_{j=1}^n \theta_{j,1}\cdot e^{-2\pi i \theta_{j,2}\omega_k} \cdot e^{-2\pi^2\alpha^2\omega_k^2}.
\end{align}
The low- and high-resolution sampling operators, $G_L$ and $G_H$, are defined by (\ref{def: GL}) and (\ref{def: GH}) respectively. 
Consequently, the noisy low-resolution measurement can be expressed as 
\begin{align}
y_k = g_k + W_k = \sqrt{2\pi\alpha^2}\cdot \sum_{j=1}^n \theta_{j,1}\,
       e^{-2\pi i \theta_{j,2}\omega_k}\,
       e^{-2\pi^2\alpha^2\omega_k^2}
       + W_k, 
       \quad k=-K_L,\ldots,K_L .
\end{align}
where $W_k$ is the noise term with $|W_k|<\sigma$. We assume that $2K_L+1\ge 2n$.\\

Applying Theorem \ref{thm: abstract stability} and Theorem \ref{cor: repre} to this model, we have the following stability estimate.
\begin{thm}{\label{thm: stability gauss}}
    For any given $\theta^* = (\theta_{1,1}^*,\cdots,\theta_{n,1}^*,\theta_{1,2}^*,\cdots,\theta_{n,2}^*)\in\Theta$, let $\Delta = \frac{1}{2}\cdot \min_{p\ne q} d_{\mathbb T}(\theta_{p,3}^*,\theta_{q,3}^*)$. We assume that $\theta^*_{j,1}\ne 0$ for $j=1,\cdots,n$ and $\Delta >0$. Let $U = \prod_{j=1}^n \bra{B(\theta_{j,1}^*,\frac{|\theta_{j,1}^*|}{2})\cap I_1}\times\prod_{j=1}^n B(\theta_{j,3}^*,\Delta)$. 
    Then, there exists $C_{U}>0$ such that for any $\theta,\theta'\in \Theta$, 
    \begin{align} \label{eq-cu3}
        \|\theta-\theta'\| \le C_{U} \cdot \|\PL(\theta)-\PL(\theta')\|.
    \end{align}
    Furthermore,
    \begin{align}
        \|D\PH \|_{op} \le C',
    \end{align}
    for some constant $C'$ independent of $K_H$, and we have
    \begin{align}
        \|\PH(\theta)-\PH(\theta')\|\le C_U\cdot C' \cdot \|\PL(\theta)-\PL(\theta')\|.
    \end{align}
    As a consequence, if $\hat{\theta}\in U$ is a $(\Theta, \sigma)$-admissible parameter, then 
    \begin{align}
        \| \PH(\hat{\theta})-\PH(\theta^*) \| < 2C_U\cdot C'\cdot\sigma.
    \end{align}
\end{thm}

\begin{rmk}
    Due to the smoothness of the signal profile, its Fourier transform decays rapidly in the frequency space. We can observe that the extrapolation is stable in the frequency domain.
\end{rmk}

A precise characterization of the constant $C_U$ in (\ref{eq-cu3}) for the general case is beyond the scope of this paper. Below, we present a simple example illustrating its dependence on the separation distance.

\begin{prop}\label{prop4-5}
Consider two sources in the physical domain, $e^{-\frac{(x-z)^2}{2\alpha^2}}+e^{-\frac{(x-z')^2}{2\alpha^2}}$, for $z,z'\in [0,1]$ with $0<|z-z'|<\frac{1}{4K_L}$. We denote $\Delta : =|z-z'| $. Then 

\begin{align} \label{eq-kl2}
    \|\theta -\theta' \| \le \frac{1}{ C(K_L)\cdot\Delta}\|\PL(\theta)-\PL(\theta') \|,
\end{align}
for some constant $C(K_L)>0$ with $C(K_L)\rightarrow C$ as $K_L\rightarrow \infty$.
    
\end{prop}

\begin{rmk}
    The dependence on $K_L$ in (\ref{eq-kl2}) differs from that in (\ref{eq-kl1}). This observation aligns with the intuition that an FRI signal carries substantial energy at high frequencies, while a Gaussian mixture does not.
\end{rmk}


We note that there are few studies on the super-resolution problems for signals with continuous profiles in the physics space. From the results in Section \ref{sec: sta_spase}, we observe that appropriate modeling leads to a stable solution to the super-resolution problem. However, for general signals, choosing an appropriate model is challenging. Whether using a physics-based or a data-driven model remains a topic worthy of exploration.

\section{Numerical Experiments}{\label{sec: Numerical experiments}}
In this section, we conduct numerical experiments to test the numerical behavior of the proposed method on different signal models. Throughout this section, we define the signal-to-noise ratio for the low-resolution signal as 
\begin{align}\label{eq: snr}
    \operatorname{SNR} := 10\cdot \log_{10} \frac{\|\operatorname{signal}\|}{\|\operatorname{noise}\|}.
\end{align}
The experiments is based on the sparse $\sigma$-compatible models, all the algorithms to solve the nonlinear least-square problem are based on the Nesterov accelerated gradient descent method.

\subsection{Point Source Model}\label{subsec: num point}
In this section, we conduct two groups of experiments to test the numerical behavior of the proposed numerical scheme for the point source model. \\

First, we test the stability. We fix $K_L=10$, then the corresponding Rayleigh length is given by $RL=\frac{1}{20}$. We set 5 groups of point sources aligned in $[0,1)$ in the following way. The point sources are separated by $1 RL$ in each group, and different groups are separated by $3 RL$. We set the amplitude of the sources to follow the uniform distribution $\mathcal{U}[1,2]$, and the SNR to be around 20. We conduct $20$ random experiments where the randomness is from the amplitudes and noise. In each experiment, we pick the initial guess of the source positions by perturbing  $0.4 RL$ to the ground truth of source positions. Figure \ref{fig: point_error} shows the numerical result of the above experiments with average $\operatorname{SNR} = 19.18$. \\

We use the next experiment above to visualize the resolution-enhanced signal in the physics domain. We choose one realization from the random experiments. For given super-resolution factors $SRF=10,20$, we first extrapolate the high-frequency data according to the reconstructed source positions and amplitudes and then calculate the signal profile in the physics domain by inverse FFT (iFFT), the result is shown in Figure \ref{fig: point_srf}.\\
\begin{figure}[htb]
    \centering
    \subfloat[Reconstruction Error]{
        \includegraphics[width=0.45\textwidth]{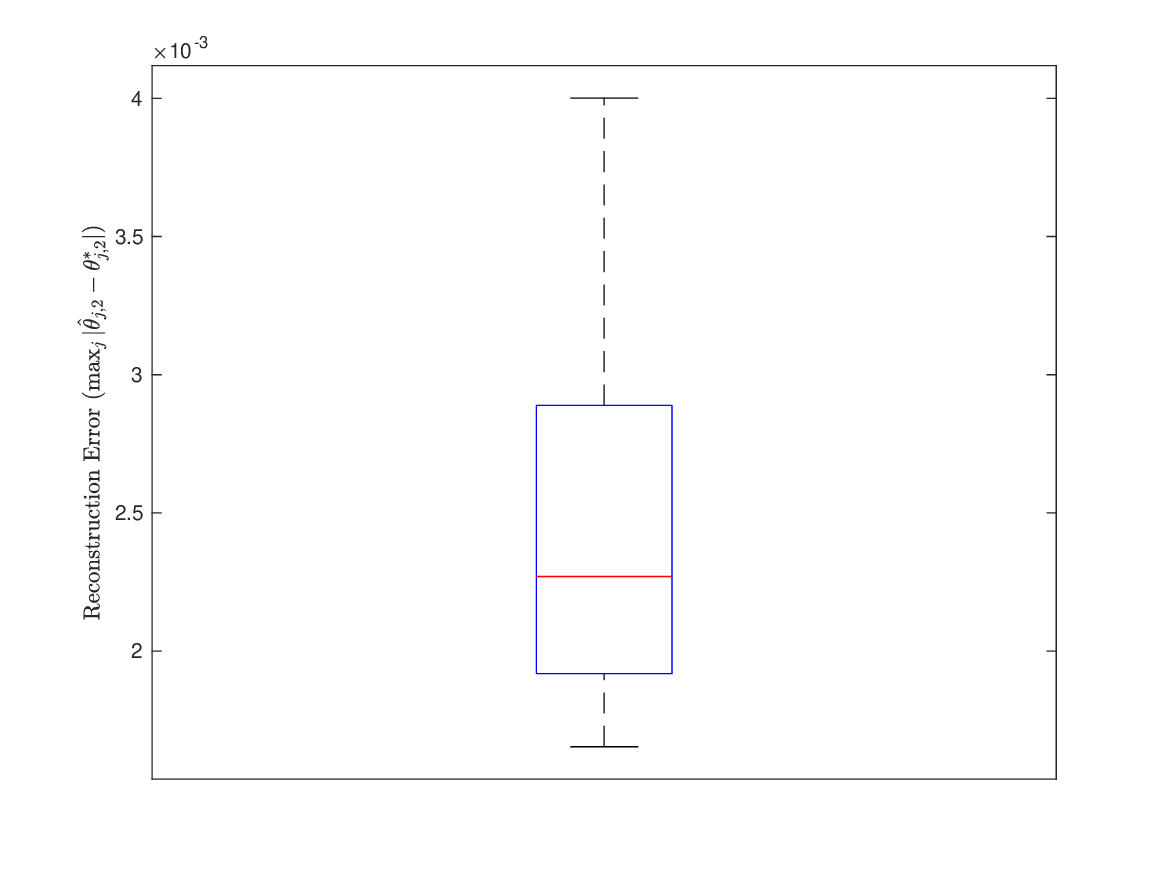}
        \label{fig: point_error}}
    \hfill
    \subfloat[Original Signal]{
        \includegraphics[width=0.45\textwidth]{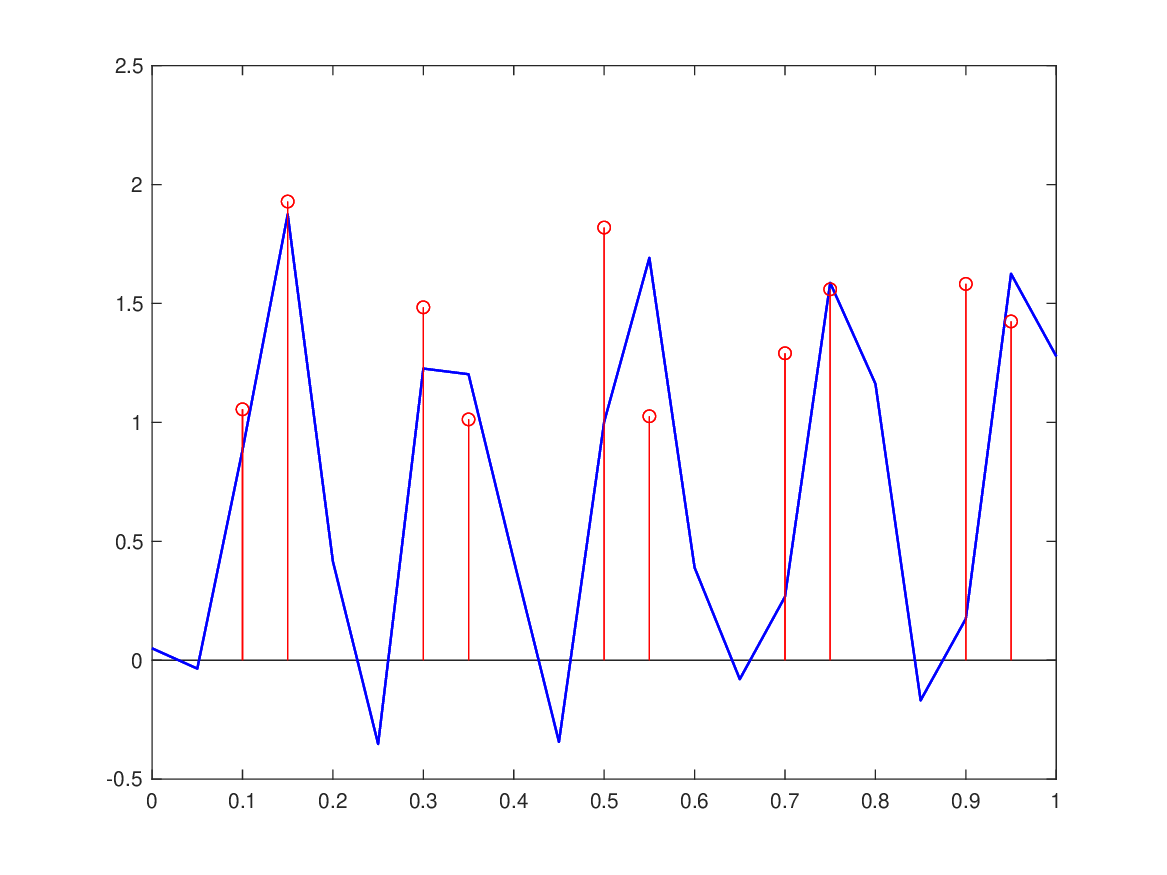}
        \label{ori}}
    \\
    \subfloat[SRF=10]{
        \includegraphics[width=0.45\textwidth]{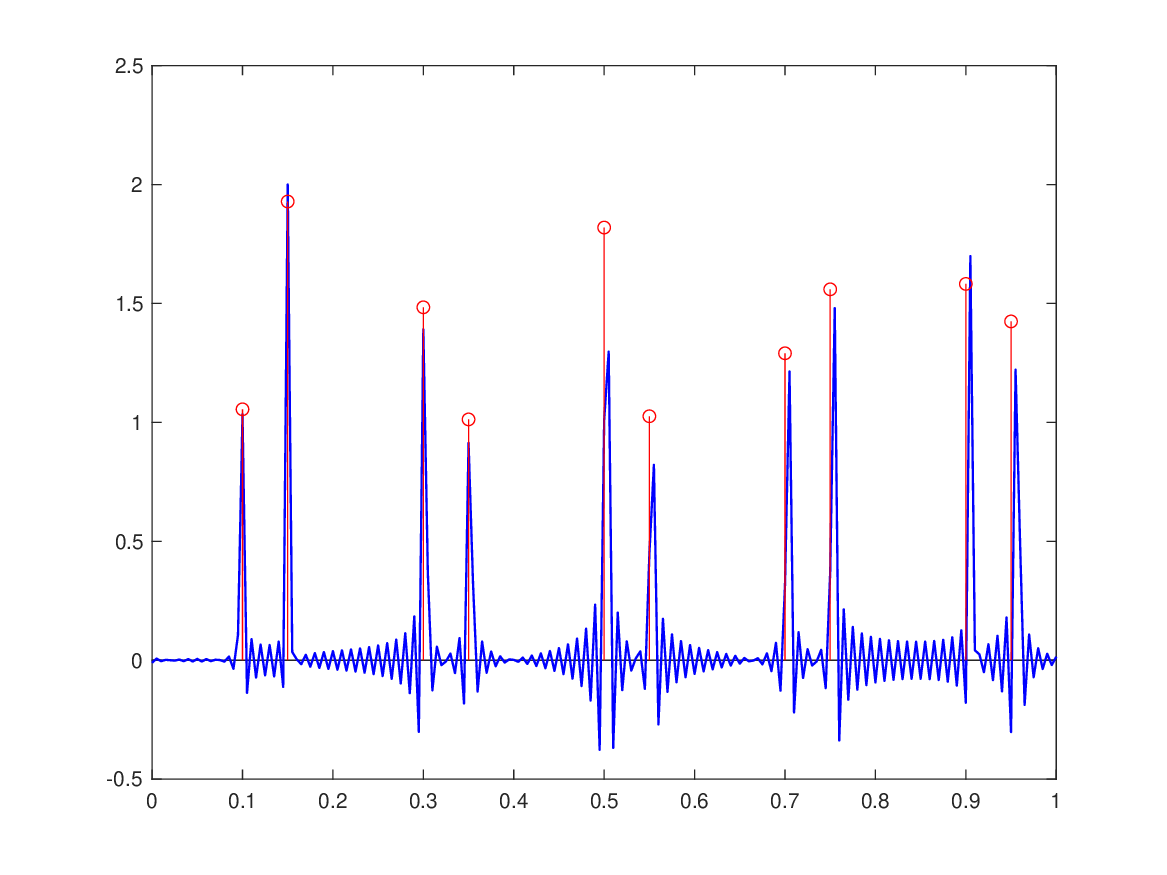}
        \label{srf10}}
    \hfill
    \subfloat[SRF=20]{
        \includegraphics[width=0.45\textwidth]{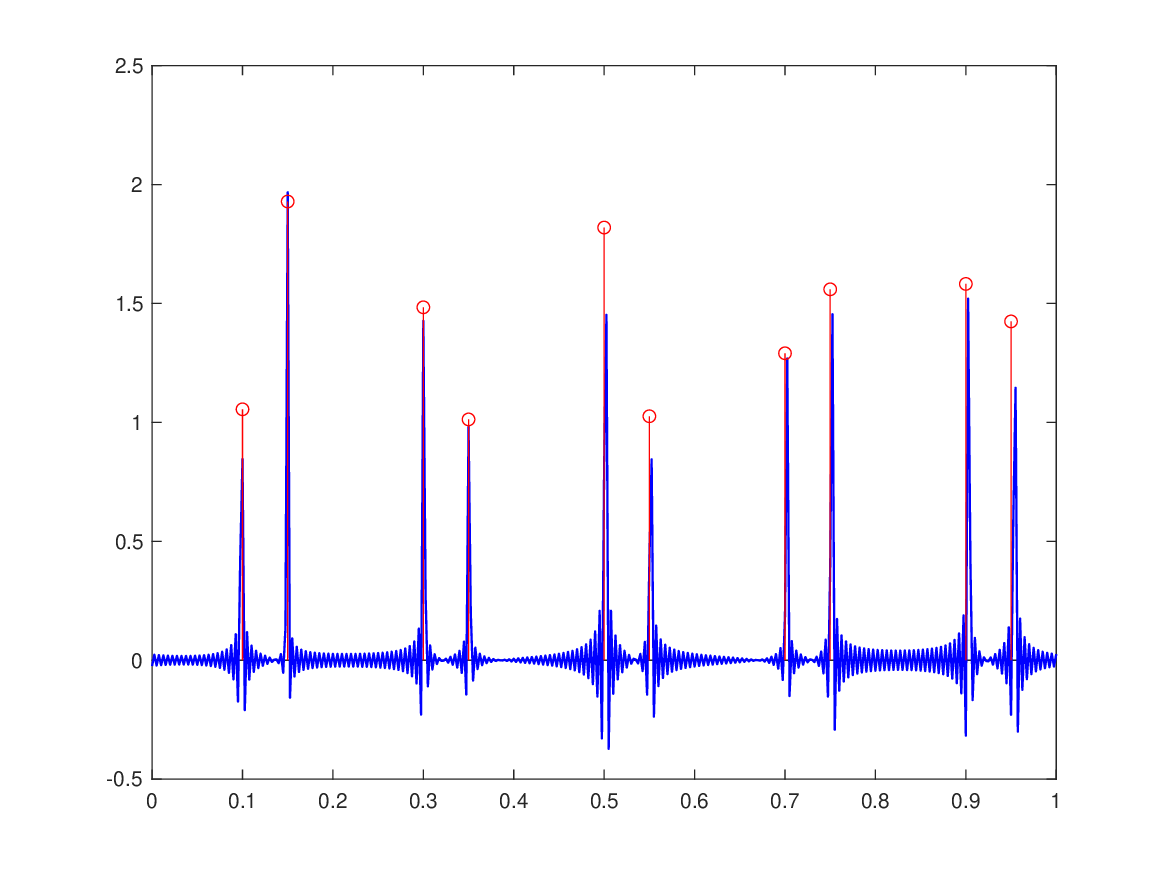}
        \label{srf20}}
    \caption{(a) Boxplot of point source position reconstruction error. (b-d) Original and resolution-enhanced signals in the physics domain. The red line represents the ground truth of the point source. The blue line shows the signal profile calculated by iFFT using the original/extrapolated Fourier data. The SNR of the experiment is $20.08$.}
    \label{fig: point_srf}
\end{figure}

Second, we demonstrate that the proposed method does not need the separation condition for the noiseless source reconstruction. We fix $K_L=5$, then the corresponding Rayleigh length is given by $RL=\frac{1}{10}$. We set two point sources in $[0,1)$ with separation distance $\frac{1}{100} RL$, and set the amplitude of the sources to follow the uniform distribution $\mathcal{U}[1,2]$. We pick the initial guess of the source position as the ground truth with a perturbation of half the separation distance. We stop the algorithm when the residue is at the level $\Ord\bra{10^{-7}}$. Figure \ref{fig: point_free} shows that the proposed method can distinguish the two point sources and give a good estimation.\\
\begin{figure}[htb]
    \centering
    \includegraphics[width=0.5\textwidth]{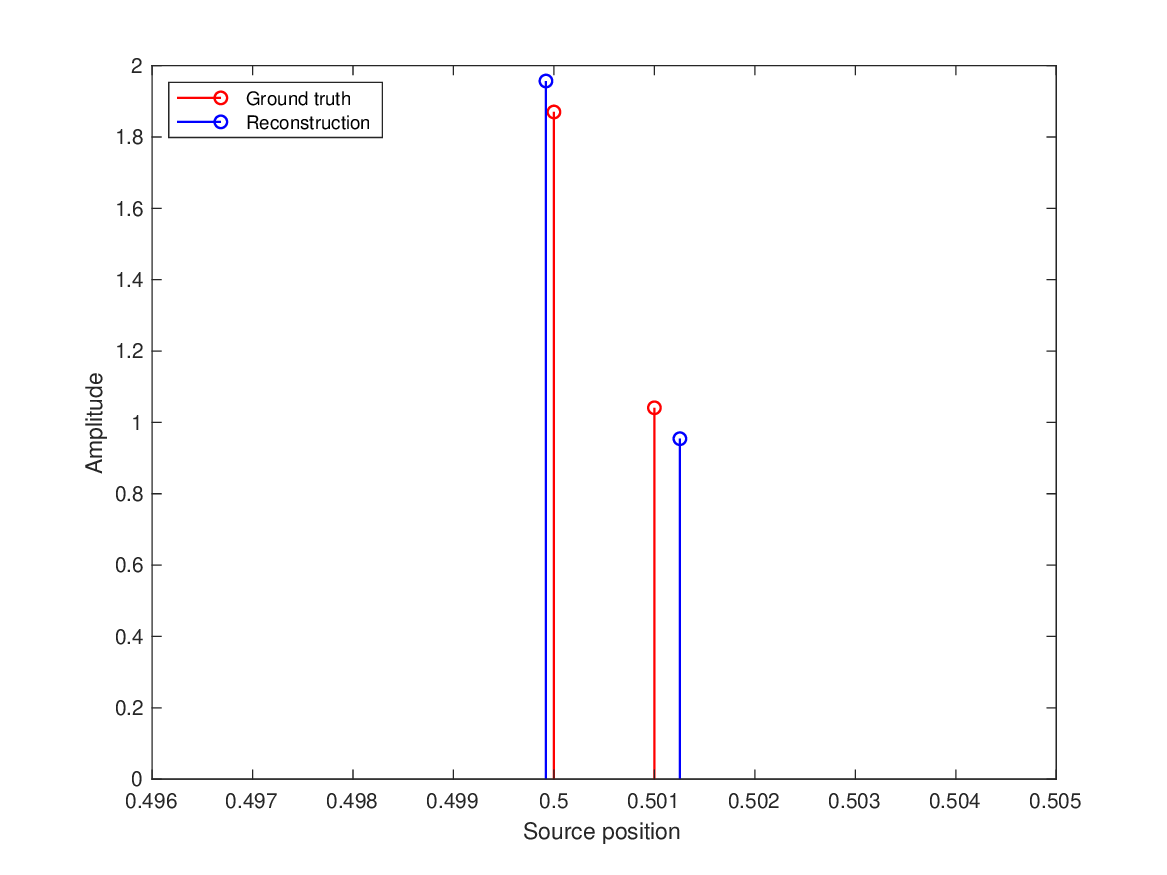}
    \caption{Reconstruction for closely positioned point sources.}
    \label{fig: point_free}
\end{figure}

\subsection{Signals with Finite Rate of Innovation} \label{subsec-7-2}
In this section, we conduct experiments on the proposed numerical scheme for signals with a finite rate of innovation. \\

In the numerical experiment, we fix $K_L=10$. The corresponding Rayleigh length is $RL=\frac{1}{20}$. The noiseless signal has the form 
\begin{align}
    \psi(x) = \sum_{j=1}^5 a_j \delta_{x_j} + \sum_{j=1}^2 b_j \delta_{y_j}'+c\delta_{z}'',
\end{align}
where $(x_1,\cdots,x_5)=(0.1,0.15,0.45,0.55,0.9)$, $(y_1,y_2)=(0.7,0.8)$, $z=0.3$. Thus, the separation distance between different sources ranges from $1RL$ to $3RL$. We call the source having the form $\delta_{x_j}$ as the monopole source, having the form $\delta'_{y_j}$ as the dipole source and $\delta''_{z}$ as the quadrupole source. We choose amplitude $a_j\sim\mathcal{U}[1,2]$ for monopole sources. We choose $b_j$'s and $c$ by ensuring that signals generated by different types of sources have comparable low-resolution signals in $\ell_2$ norm. 
We conduct 20 random experiments. In each experiment, we pick the initial guess of the source positions by perturbing  $0.4 RL$ to their ground truth. Figure \ref{fig: error_FRI} shows the reconstruction result for the experiment with average $\operatorname{SNR}=32.03$.\\

In Figure \ref{fig: error_FRI}, we observe that the absolute position reconstruction error of the quadrupole source is relatively small. This is because the loss function, especially its high-frequency part, is more sensitive to the higher-order poles.\\

To visualize the resolution-enhanced signal in the physics domain, we pick one realization from the random experiment and introduce the Dirichlet kernel 
\begin{align}
    D_{K}(x) = \sum_{k=-K}^{K} e^{2\pi ikx}. 
\end{align}
The convolution of the derivative of Dirac with the Dirichlet kernel with increasing $K_L$ leads to the significant amplification of the signal strength. More precisely
\begin{align}
    D_{K}(x) \ast \delta^{(r)}(x) = \sum_{k=-K}^{K} (2\pi ik)^r e^{2\pi ikx}.
\end{align}
Therefore different order of Diracs generate signals of different amplitudes in the physics domain. We plot these different types of signals in the physics domain in different pictures, see Figure \ref{fig: FRI_SRF}. In the Figure, the ground truth for $r\ge 1$ is generated by convoluting the ground truth derivatives of Diracs with the corresponding Dirichlet kernel.\\

\begin{figure}[!ht]
\centering
    \subfloat[Original signal]{
	\includegraphics[width=0.4\textwidth]{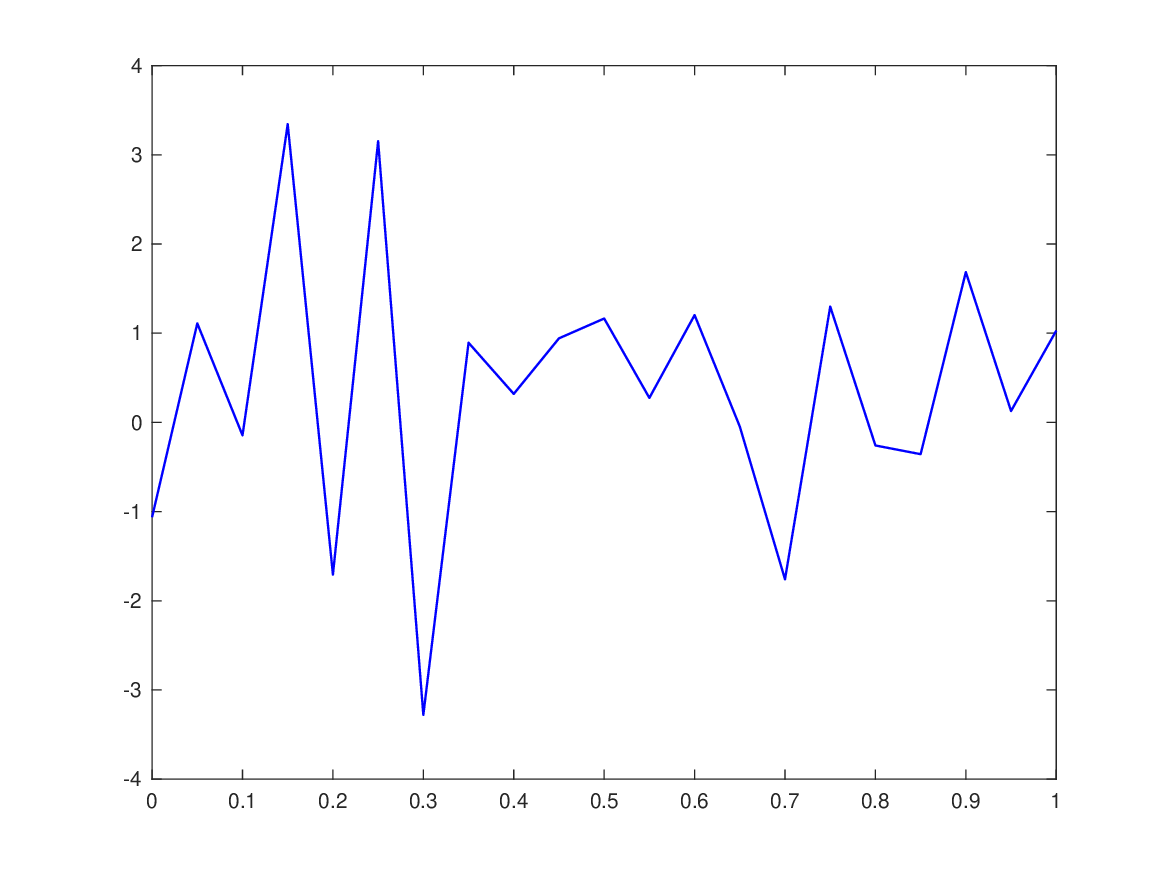}
    \label{FRI: ori}}
     \subfloat[Position reconstruction error for different types of sources.]{
	\includegraphics[width=0.4\textwidth]{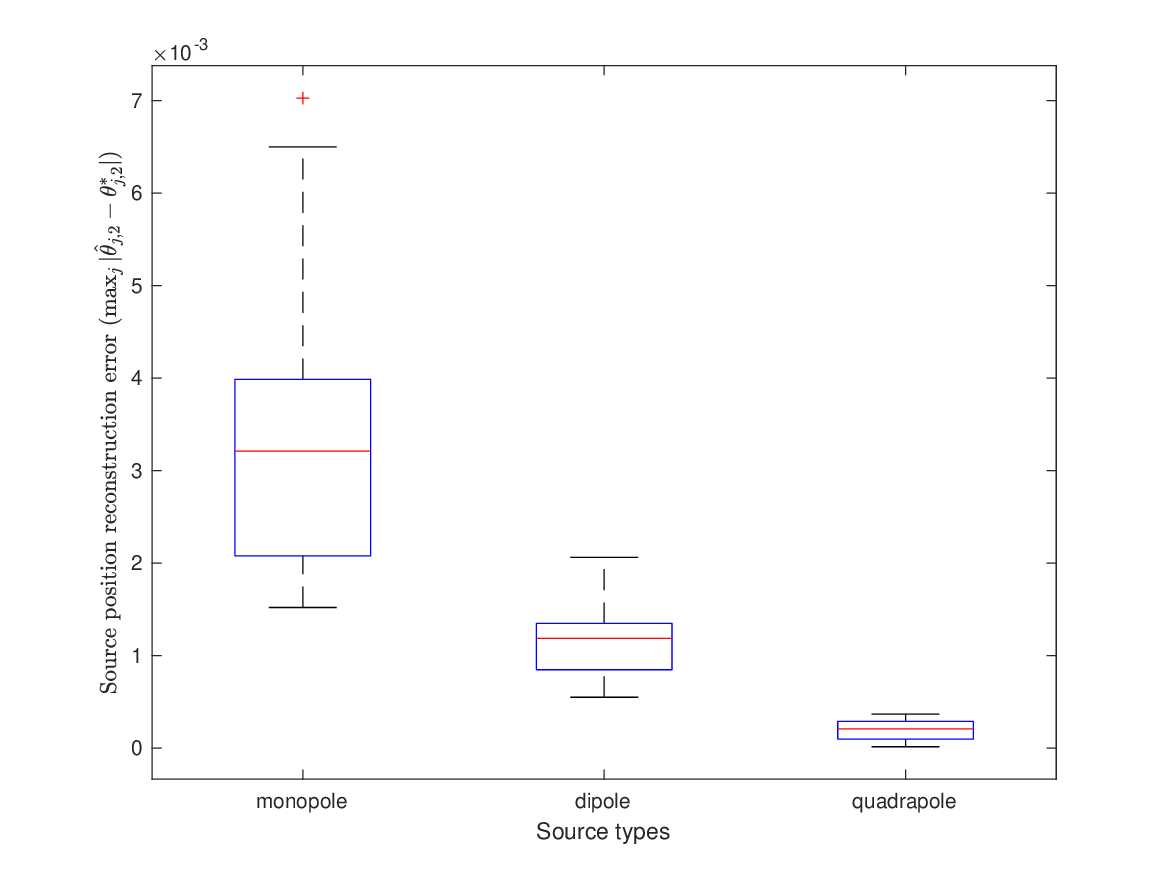}
    \label{fig: error_FRI}}
 \hspace{30mm}
    \subfloat[Monopole source, SRF=5]{
	\includegraphics[width=0.3\textwidth]{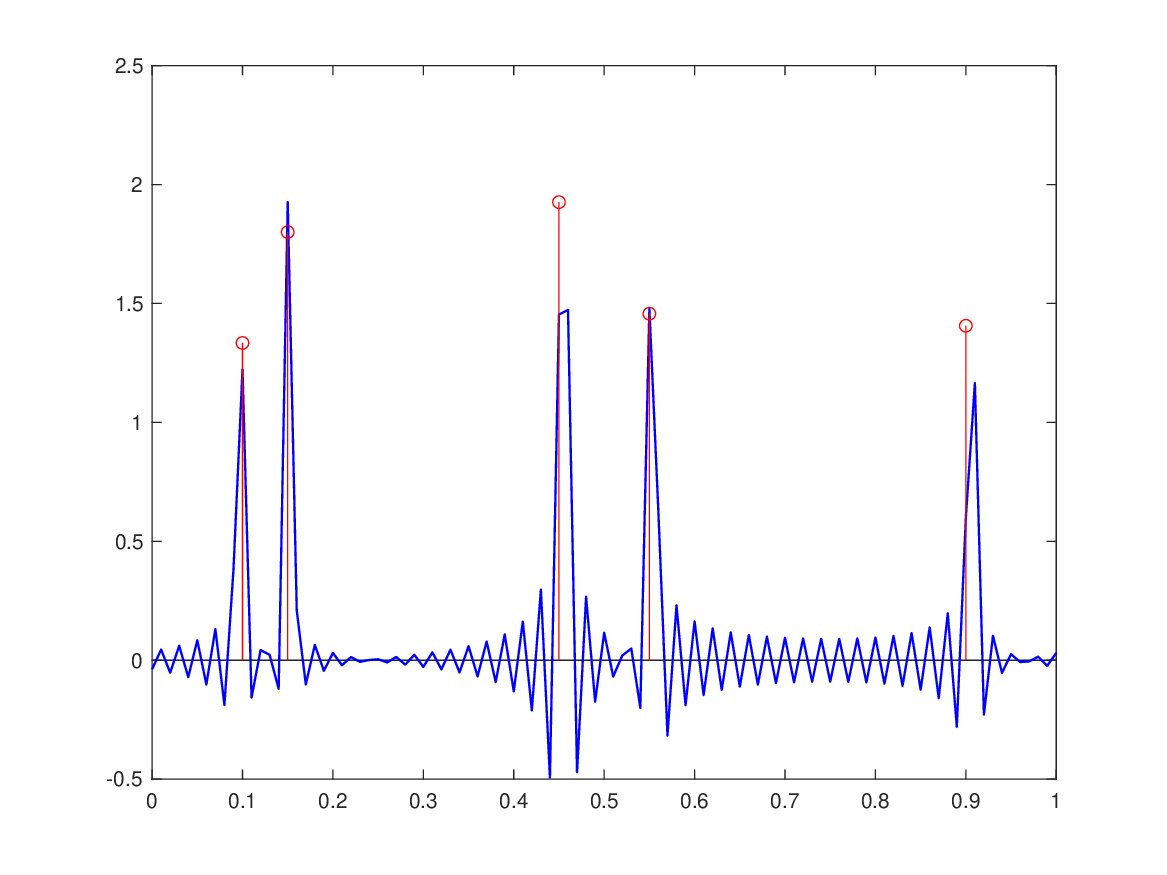}
    \label{FRI: mono 5}}
    \vspace{1mm}
    \subfloat[Dipole source, SRF=5]{
	\includegraphics[width=0.3\textwidth]{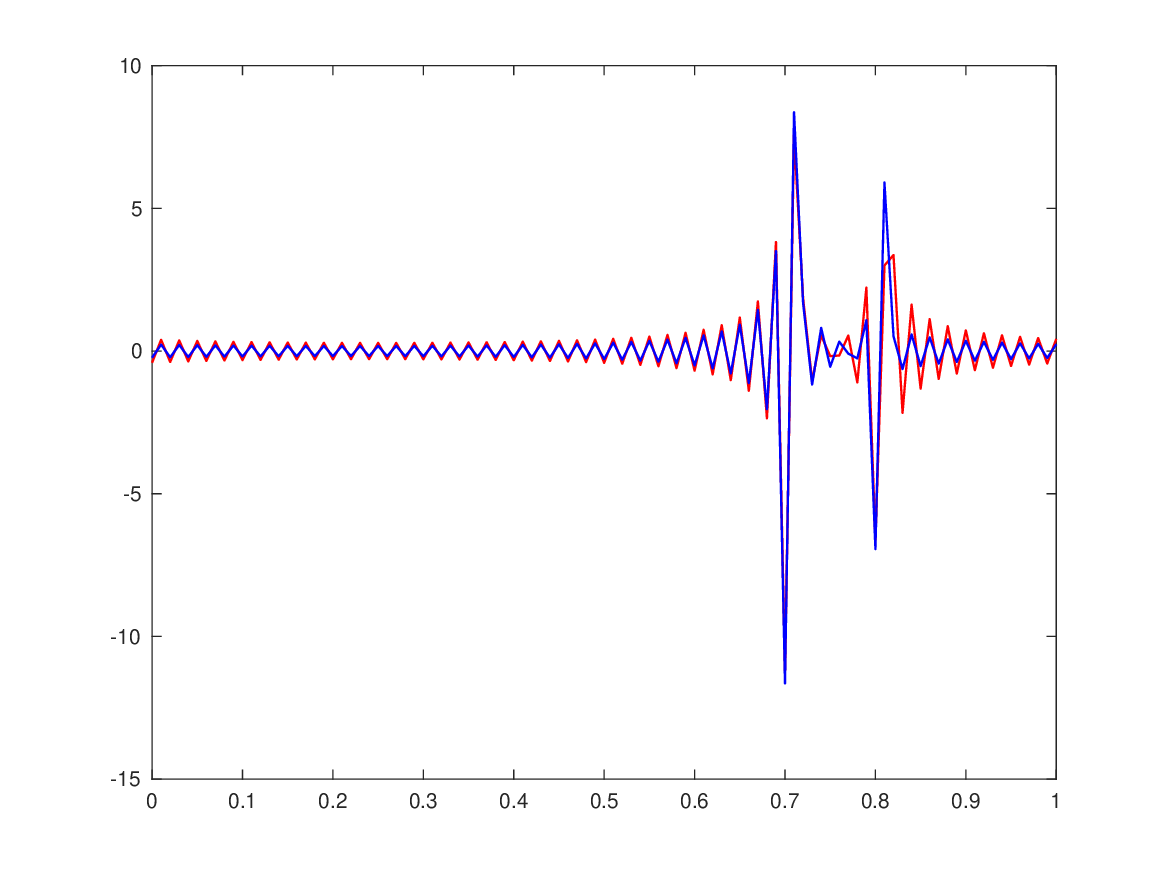}
    \label{FRI: di 5}}
    \vspace{1mm}
    \subfloat[quadrupole source, SRF=5]{
	\includegraphics[width=0.3\textwidth]{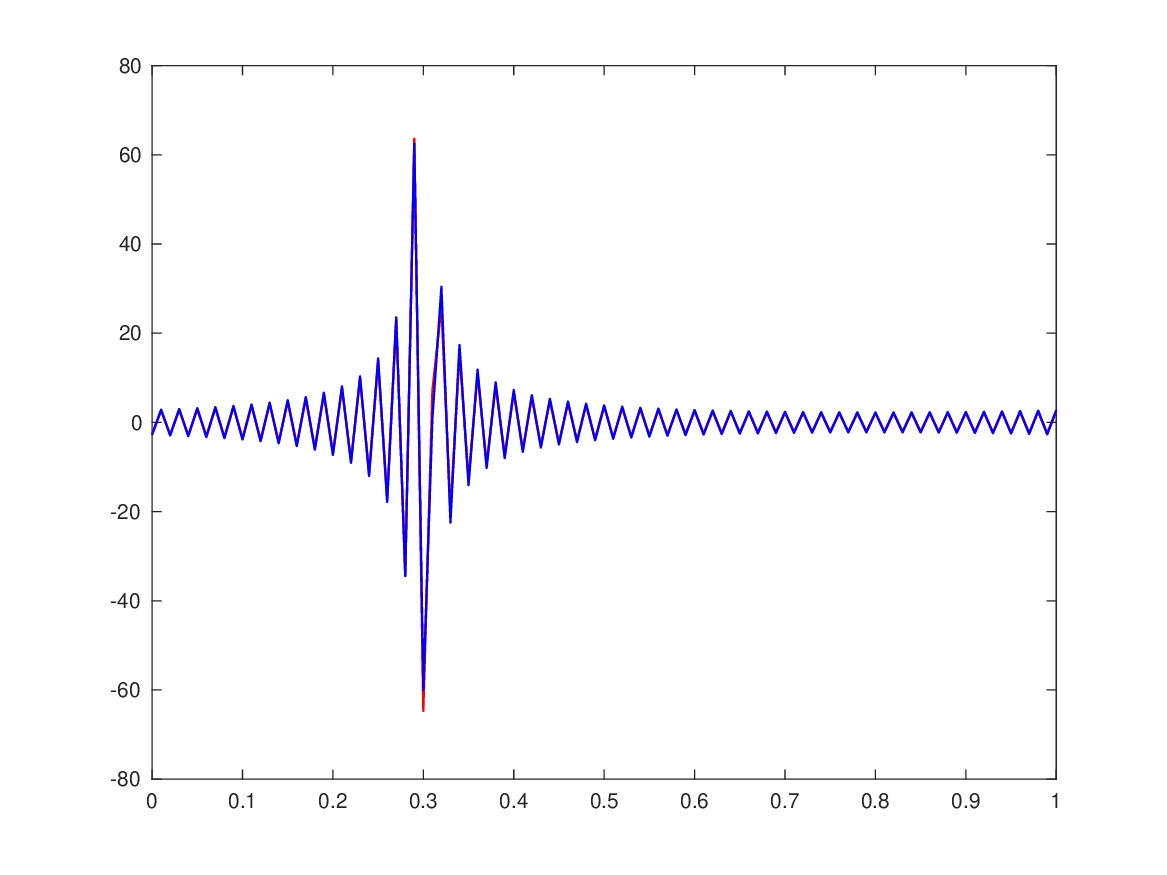}
    \label{FRI: qua 5}}

    \subfloat[Monopole source, SRF=10]{
	\includegraphics[width=0.3\textwidth]{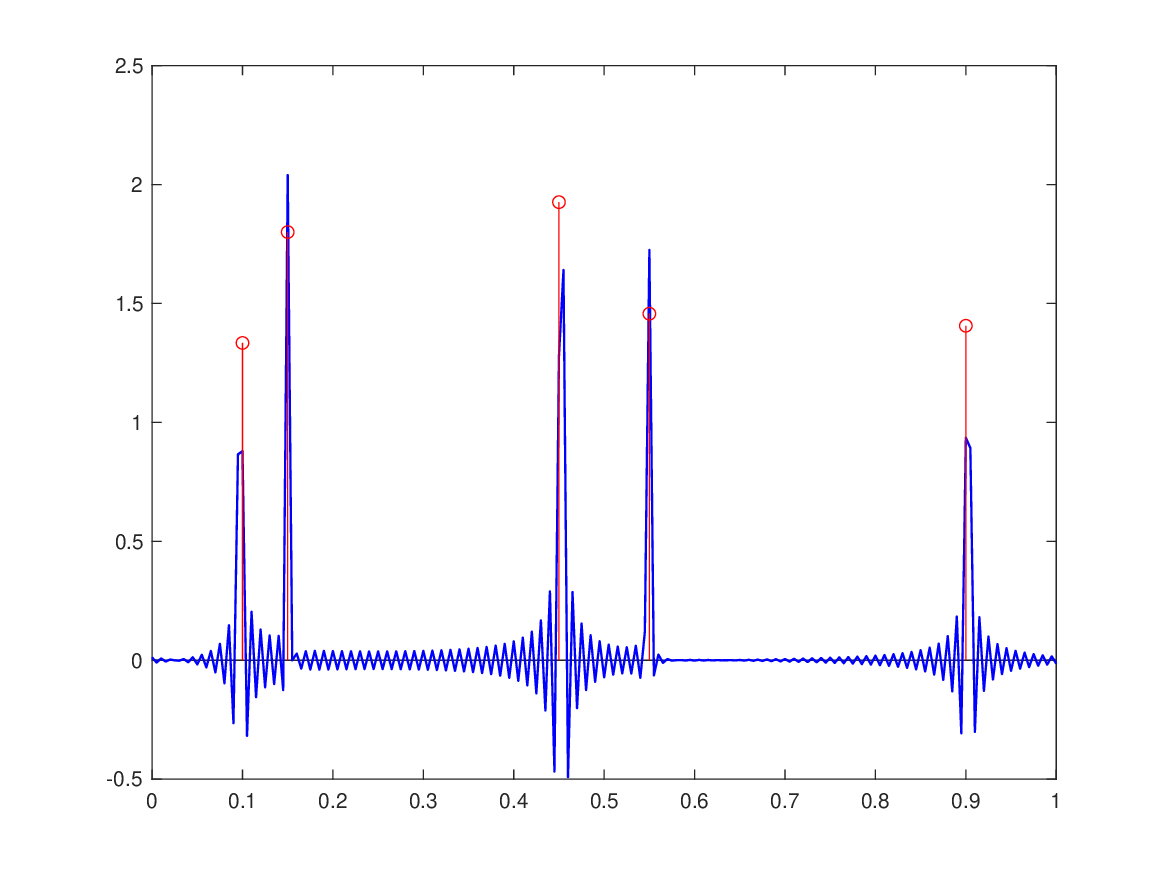}
    \label{FRI: mono 10}}
    \vspace{1mm}
    \subfloat[Dipole source, SRF=10]{
	\includegraphics[width=0.3\textwidth]{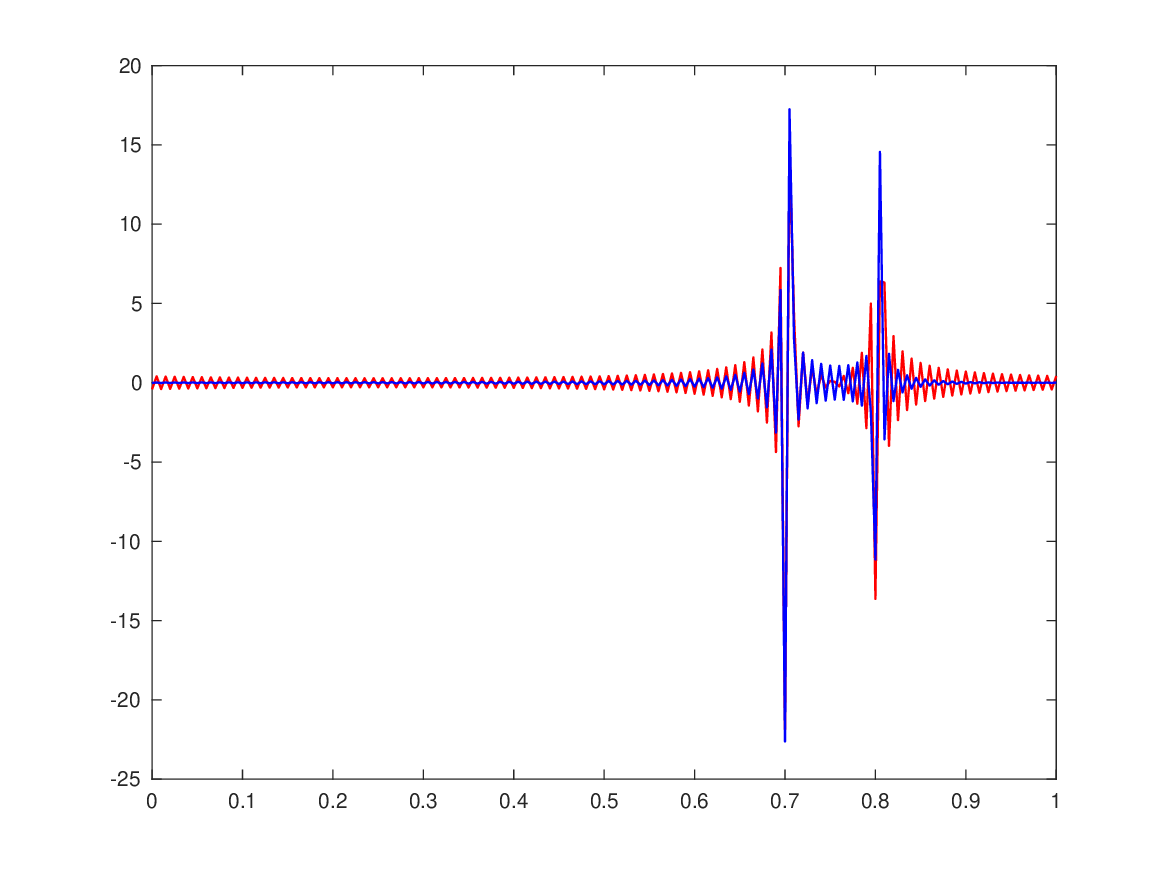}
    \label{FRI: di 10}}
    \vspace{1mm}
    \subfloat[quadrupole source, SRF=10]{
	\includegraphics[width=0.3\textwidth]{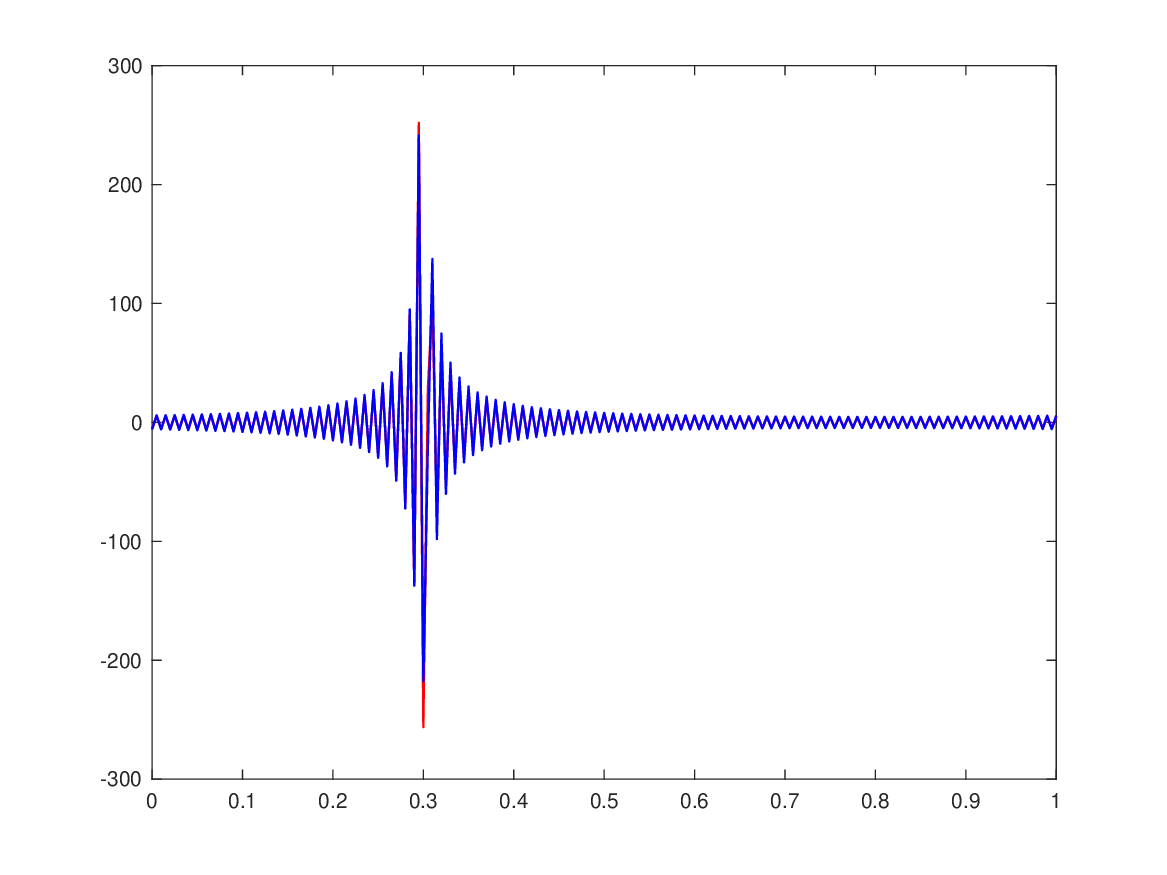}
    \label{FRI: qua 10}}
 \caption{Original signal and resolution-enhanced signals in the physics domain. The blue line in all figures shows the signal profile calculated by iFFT using the original/extrapolated Fourier data. The red line in Figure \ref{FRI: mono 5} and Figure \ref{FRI: mono 10} shows the ground truth of point sources. The red line in Figure \ref{FRI: di 5}, Figure \ref{FRI: di 10}, Figure \ref{FRI: qua 5} and Figure \ref{FRI: qua 10} shows the signal profile of the ground truth higher pole sources sampled by the corresponding Dirichlet kernel. The SNR of the experiment is $30.16$.}
\label{fig: FRI_SRF}
\end{figure}

The experiment results demonstrate a stable reconstruction of the source positions. Meanwhile, we observe that the extrapolation in the frequency domain results in reliable resolution-enhanced signals in the physics domain.

\subsection{General Signals}{\label{subsec: num general}}
We conduct numerical experiments for more complicated signals. In the physics domain, we assume that the signal is a linear combination of components having the following form 
\begin{align}
    c(x;\boldsymbol{\mu}) = e^{i(\mu_0 x^2+\mu_1 x+\mu_2)}\cdot e^{-\frac{(x-\mu_3)^2}{2\mu_4^2}},
\end{align}
where $\mu_0,\mu_1,\mu_2 \in \Real$, $\mu_3 \in [0,1)$, $\mu_4\in (0,\infty)$. Then, the signal can be written as
$$\psi(x) = \sum_{j=1}^n b_j c(x;\boldsymbol{\mu}_j).$$ Suppose we have low-frequency data in the frequency domain and aim to recover the high-frequency data to achieve super-resolution.\\

The experiment considers a signal having $4$ components with different $\boldsymbol{\mu}$'s.  Write the signal in the following equivalent form
$$\psi(x) = \sum_{j=1}^4 (\kappa_{0,j}+\kappa_{1,j} i)\cdot e^{i(\kappa_{2,j} x^2+\kappa_{3,j} x)}\cdot e^{-\frac{(x-\kappa_{4,j} )^2}{2\kappa_{5,j} ^2}}.$$\\

Notice that, different from the signal models in Section \ref{subsec: num point} and \ref{subsec-7-2}, we do not have the explicit form of the Fourier transform for the signal above. In the experiment, we set $(\kappa_{4,1},\kappa_{5,1}) = (0.2,0.02)$, $(\kappa_{4,2},\kappa_{5,2}) = (0.4,0.03)$, $(\kappa_{4,3},\kappa_{5,3}) = (0.6,0.01)$, and $(\kappa_{4,4},\kappa_{5,4}) = (0.8,0.01)$. In the physical space $[0,1)$, we setup a grid $\{x_t^{(c)}\}$, defined by $x_{t}^{(c)} = \frac{t}{127}$, $t=0,\cdots,127$, for the calculation of FFT. Using this grid, we generate $32$ noisy low-frequency data. Then, the associated grid in $[0,1)$, $\{x^{(o)}_t\}$, has step size $\frac{1}{31}$. We solve the parameters $\{\kappa_{p,q}\}$ by the low-frequency data and draw the picture of the signal on two finer grids having step size $\frac{1}{127}$ and $\frac{1}{4095}$ respectively. We pick the initial guess of $\kappa_{4,j}$ by adding or minus a constant around $0.07$ to the ground truth (noticing that the $RL=\frac{1}{32}$ for the system). Consequently, the initial guess for $\kappa_{4,1}$ has error $3.5\kappa_{5,1}$, the initial guess for $\kappa_{4,2}$ has error more than $2\kappa_{5,2}$, and the initial guess for $\kappa_{4,3}$ and $\kappa_{4,4}$ has error $7\kappa_{5,3}$.  During the optimization, we restrict $\kappa_{4,j}$ to $ (0,1)$ by re-initialization if $\kappa_{4,j}\notin(0,1)$ in some step. We visualize the signals in the physics domain, see Figure \ref{fig: gauss}. The original signal is calculated from the iFFT of the noisy low-frequency samples. The resolution-enhanced signals are calculated by the interpolation of the recovered signal profile. The experiment is conducted under $\operatorname{SNR}=11.35$. \\

\begin{figure}[!ht]
\centering
    \subfloat[Absolute value, Original]{
	\includegraphics[width=0.3\textwidth]{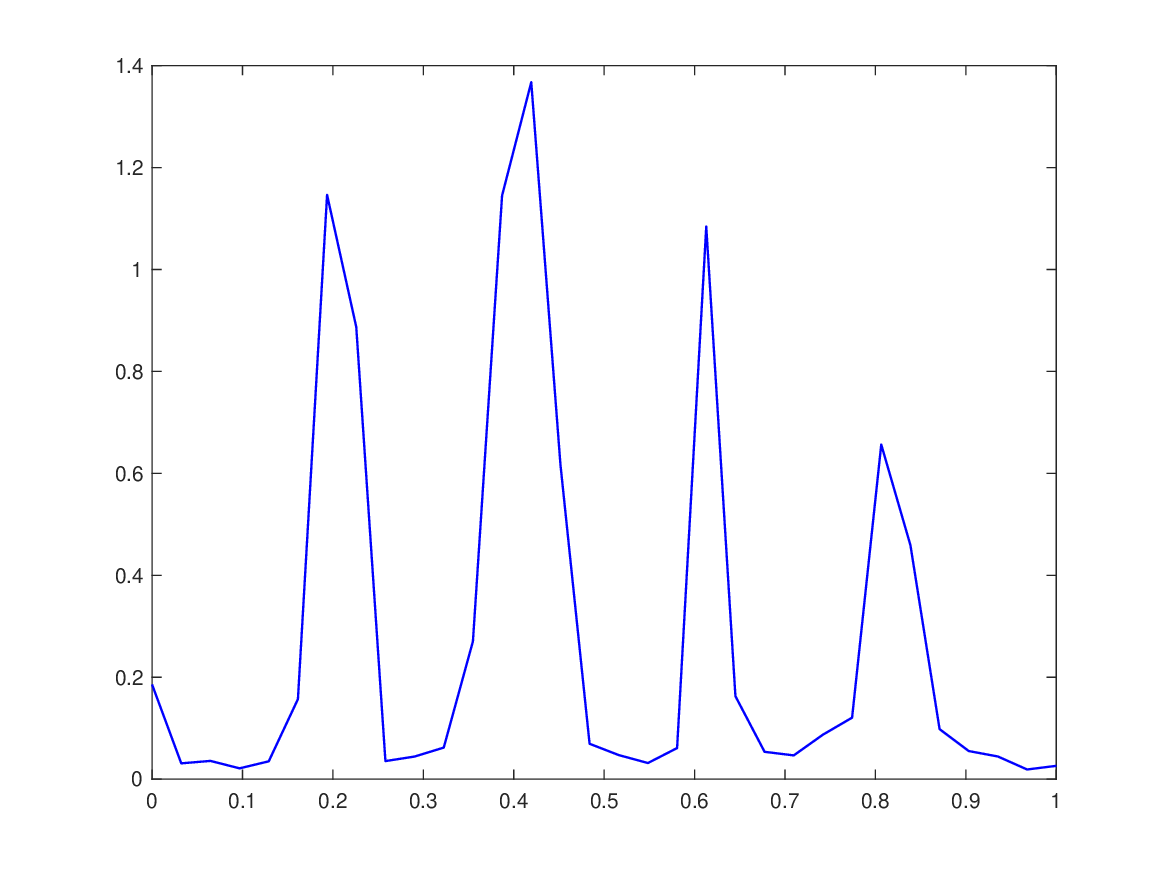}}
    \vspace{1mm}
    \subfloat[Real part, Original]{
	\includegraphics[width=0.3\textwidth]{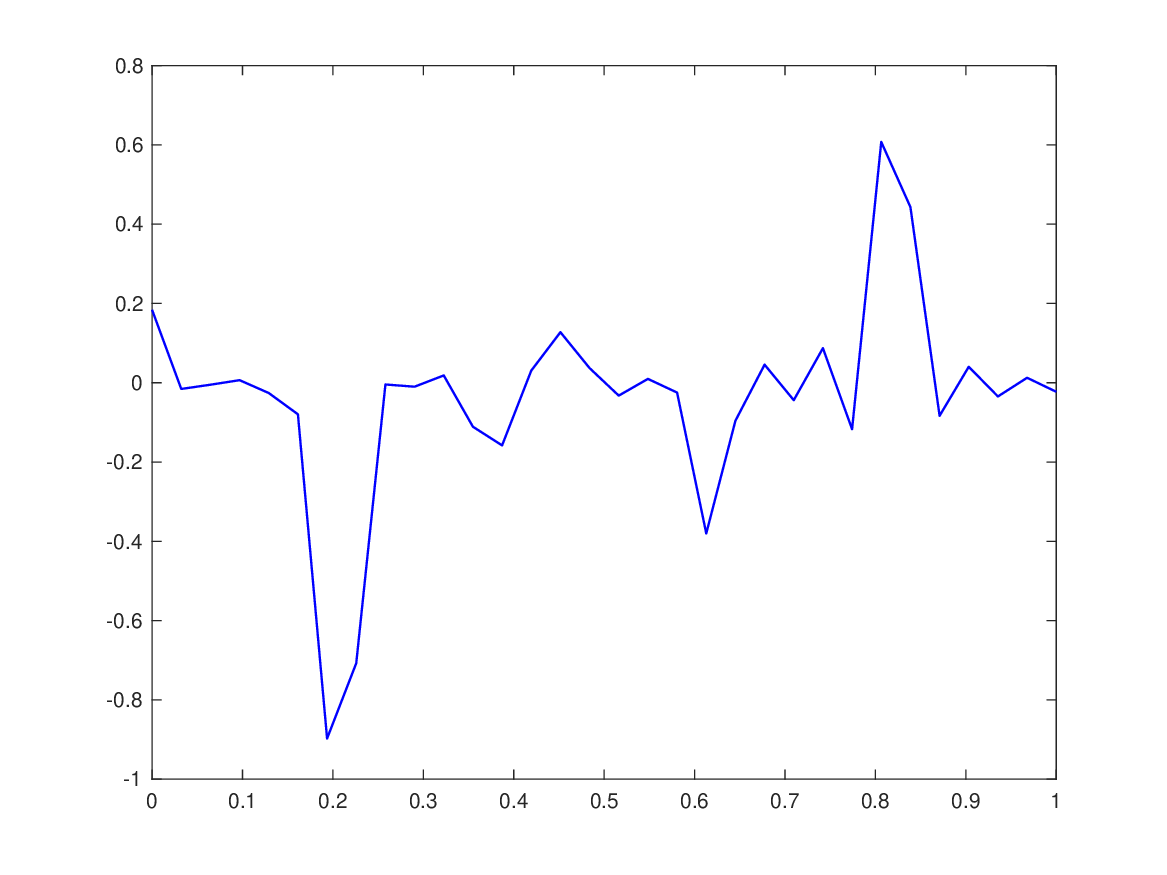}}
    \vspace{1mm}
    \subfloat[Imaginary part, Original]{
	\includegraphics[width=0.3\textwidth]{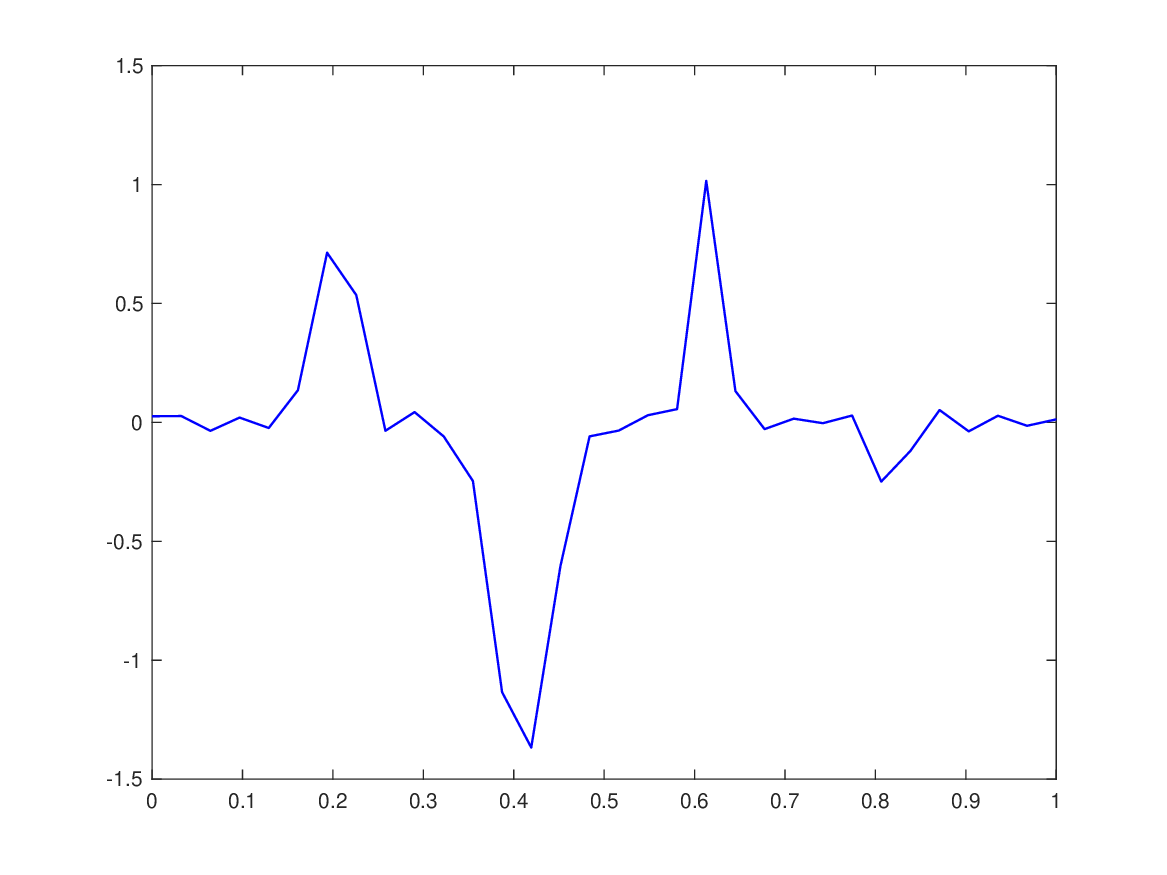}}
 
    \subfloat[Absolute value, SRF=4]{
	\includegraphics[width=0.3\textwidth]{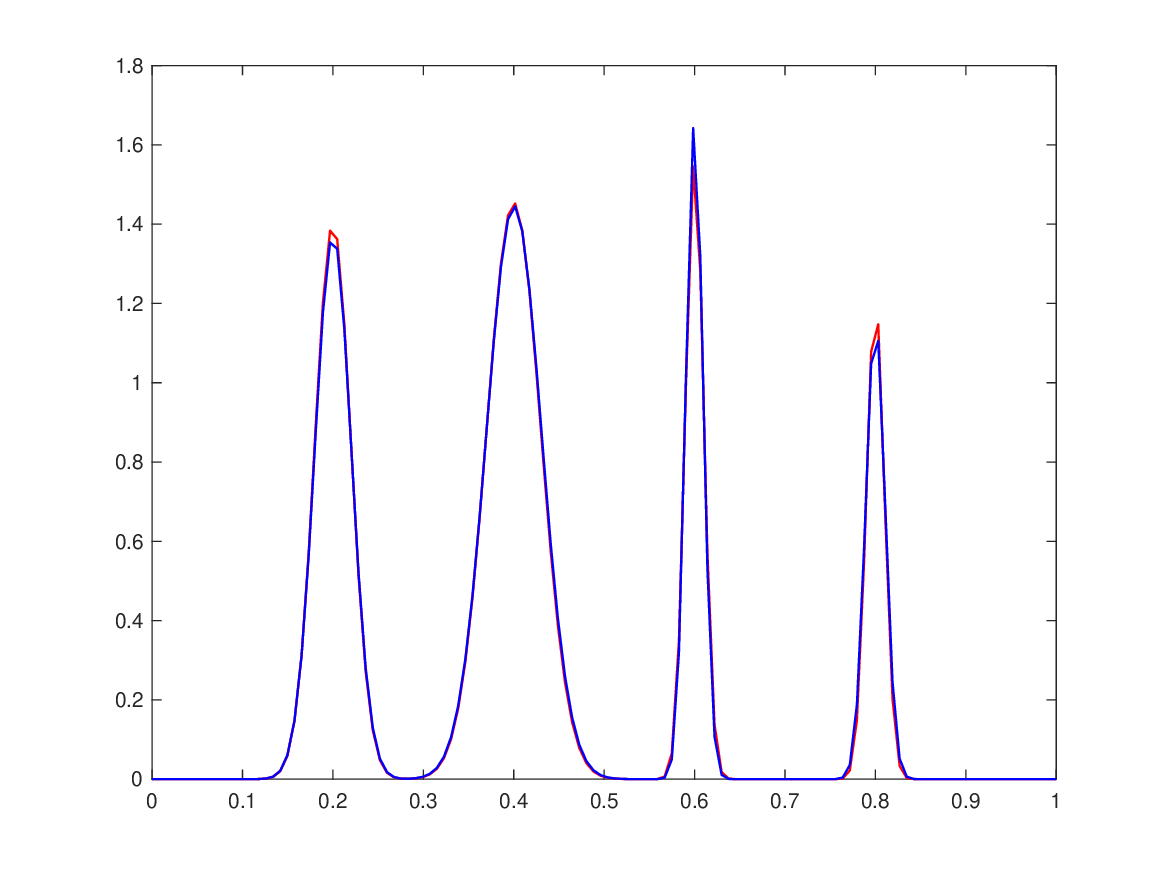}}
    \vspace{1mm}
    \subfloat[Real part, SRF=4]{
	\includegraphics[width=0.3\textwidth]{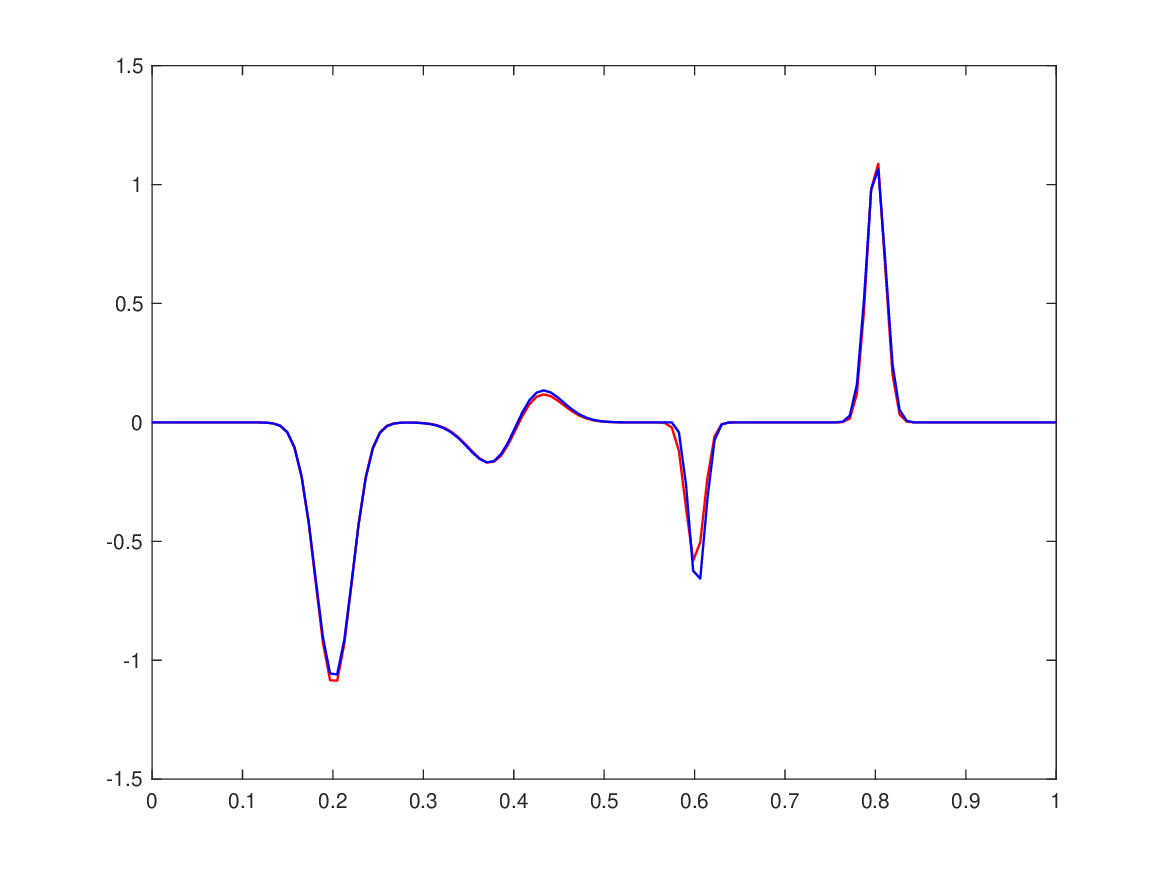}}
    \vspace{1mm}
    \subfloat[Imaginary part, SRF=4]{
	\includegraphics[width=0.3\textwidth]{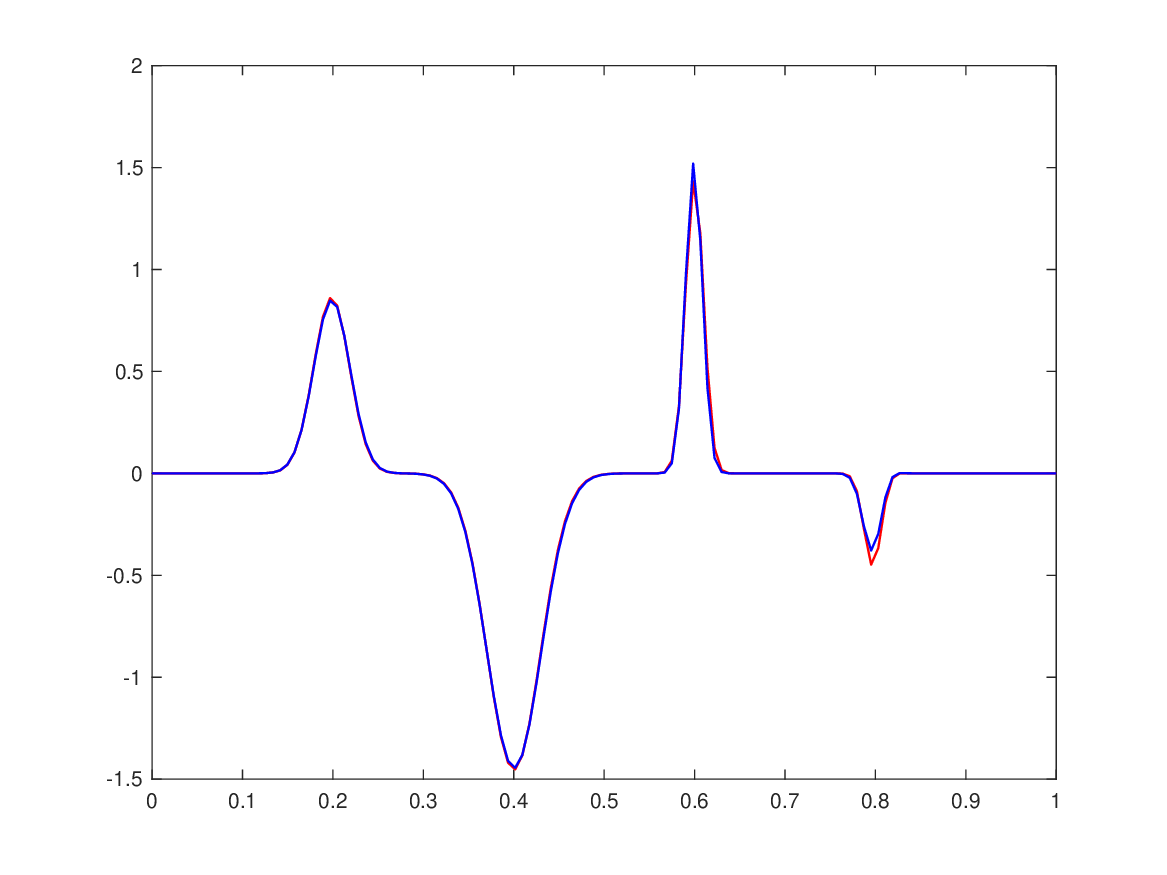}}

    \subfloat[Absolute value, SRF=128]{
	\includegraphics[width=0.3\textwidth]{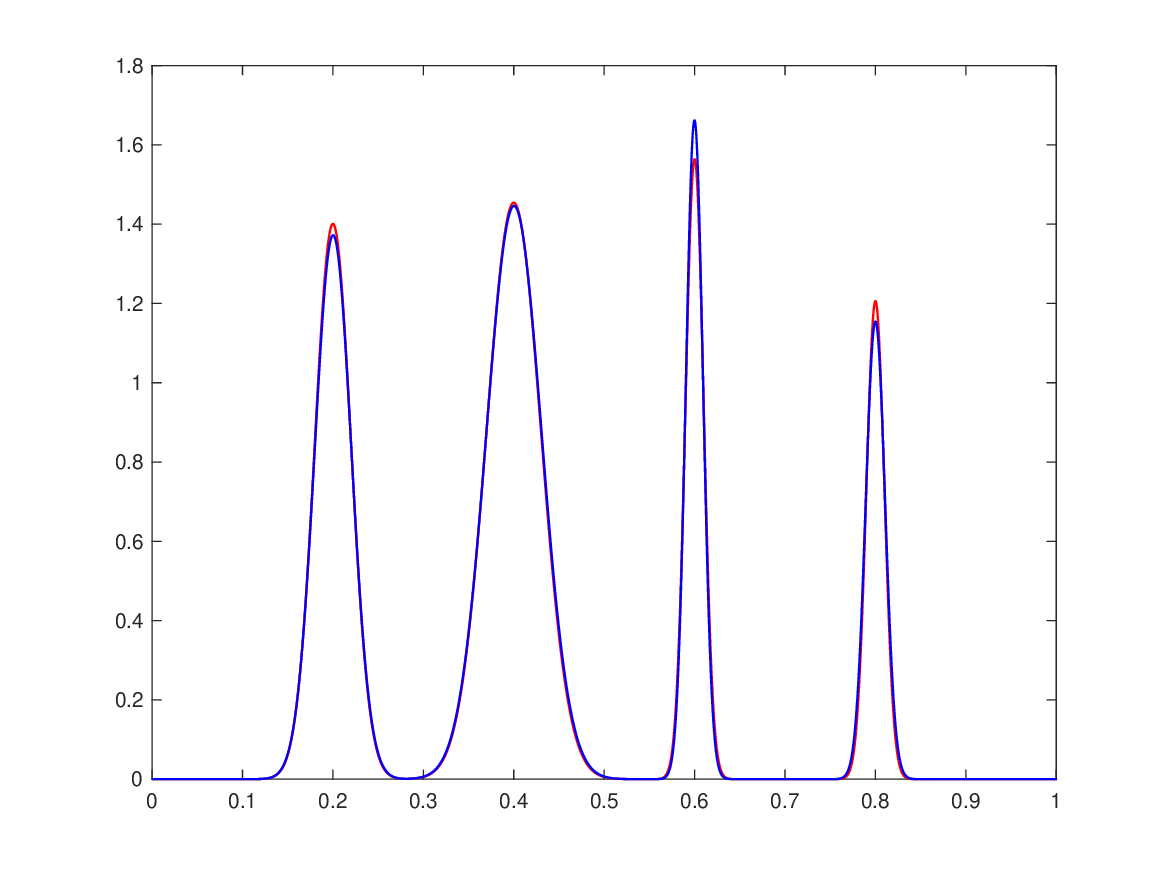}}
    \vspace{1mm}
    \subfloat[Real part, SRF=128]{
	\includegraphics[width=0.3\textwidth]{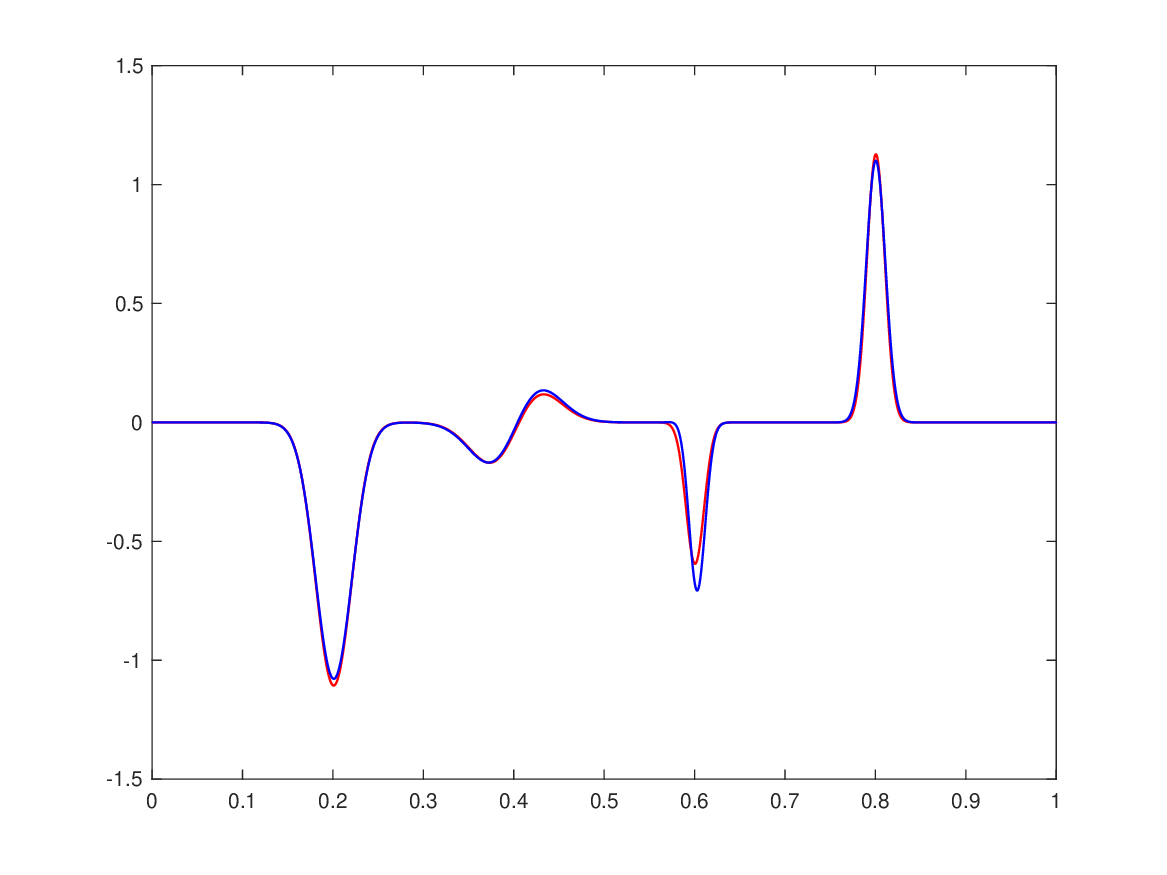}}
    \vspace{1mm}
    \subfloat[Imaginary part, SRF=128]{
	\includegraphics[width=0.3\textwidth]{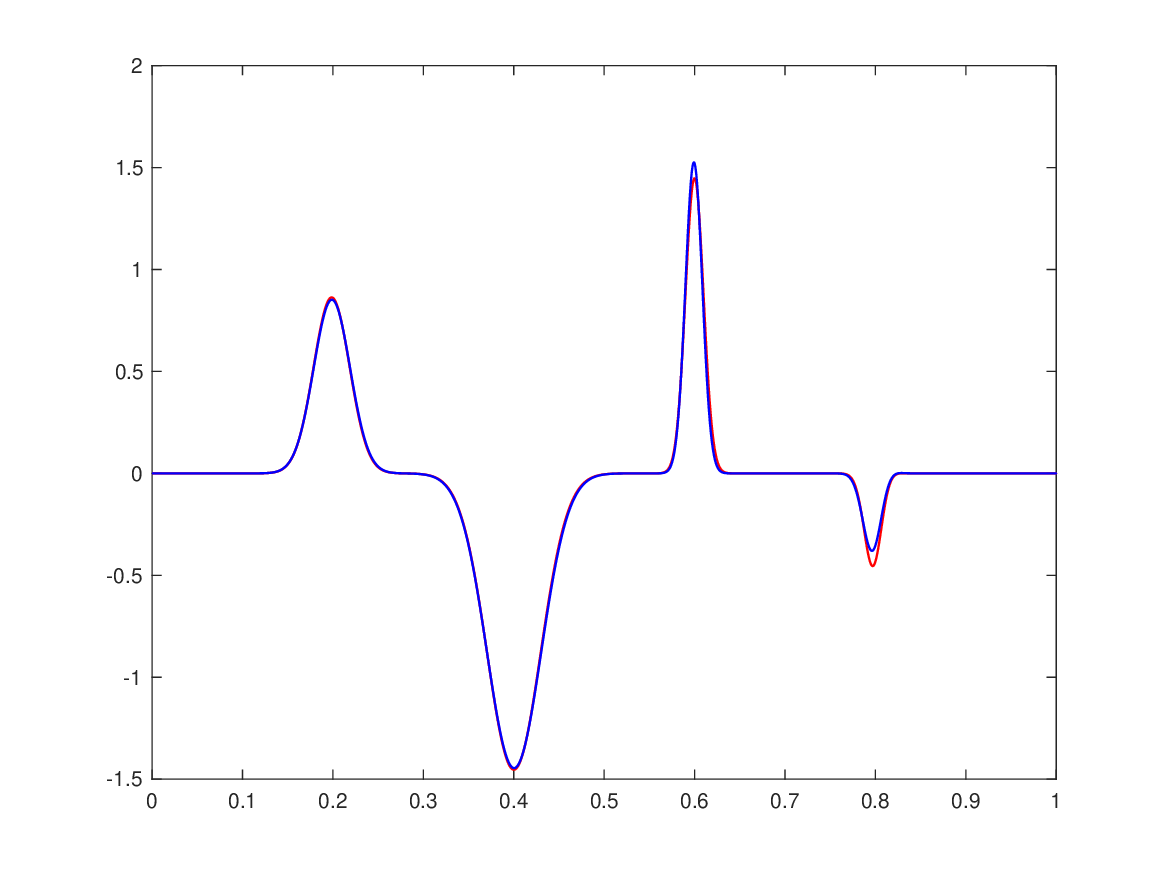}}
 \caption{Original signals and resolution-enhanced signals in the physics domain. The red line represents the ground truth, and the reconstructed ones are shown in blue. We define the super-resolution factor (SRF) in the physics domain as the quotient of the two grid point numbers. The first row shows the original signal. The second and third rows show the resolution-enhanced signal with $\operatorname{SRF}=4$ and $10$, respectively. The first column is the absolute value of the signal profile, and the second and third columns are the real and imaginary parts of the signal, respectively. The SNR of the experiment is $11.35$.}
\label{fig: gauss}
\end{figure}

Notice that using a finer grid of the physics domain to calculate FFT gives a more accurate approximation of the Fourier transform but at the expense of higher computational cost since the Fourier transform is excuted in each iteration for the optimization problem.\\

Finally, we point out that the signal profile used in this experiment is smooth and thus has a rapidly decaying Fourier transform. Therefore, the high-frequency data is very noisy, and the low-frequency data plays an important role in the parameter estimation. 

\section{Extension}\label{sec: extension}
In this section, we discuss several extensions of the proposed framework to other problems with similar structures.

\subsection{Data Completion}
For a typical data completion problem, the sampling grid is slightly different from the one used in the super-resolution problem. We denote the full grid as $\mathcal{M}_F=\{\omega_k\}_{k\in\Lambda}$ and the partial grid as $\mathcal{M}_P=\{\omega_k\}_{k\in\Lambda'}$, where $\Lambda'\subset\Lambda$. For given data space $\mathcal{M}$, we can similarly define the partial sampling operator $G_P$ and full sampling operator $G_F$ similar to low- and high-resolution sampling operators, respectively. Following the routine as in Section \ref{subsec: framework}, the model-based data completion framework can be developed. See the following Figure \ref{fig: model-based DC} for illustration. In the figure, $\mathcal{Q}$ is the downsampling operator and $\mathcal{L}$ sampling lifting operator satisfying the conditions (\ref{commu:Q}) and (\ref{commu:L}). Further, once the suitable modeling pair $(\Theta,\mathcal{P})$ is determined, the numerical methodology and theoretical estimates also apply to the model-based data completion framework.
\begin{figure}[htb]
    \centering    
    \begin{tikzcd}[row sep=small,column sep = large]
              &        &\mathcal{H}_P \arrow[dd,shift left=1.5,red,"\mathcal{L}"] \arrow[from=dd,"\mathcal{Q}"]\\ 
        \Theta\arrow[urr,bend left=15, "\mathcal{P}_P"] 
        \arrow[drr, bend right=15,"\mathcal{P}_F",swap]\arrow[r,"\mathcal{P}"]&
        \mathcal{M}\arrow[ur,"G_P"]\arrow[dr,"G_F",swap]&         \\ 
                    &                &\mathcal{H}_F 
    \end{tikzcd}
    \caption{Model-based data completion framework}
    \label{fig: model-based DC}
\end{figure}
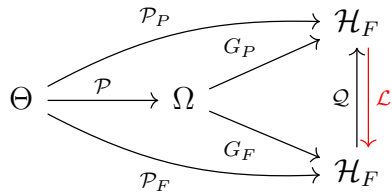
\subsection{Deep Learning}

As seen in the previous sections, to achieve stable super-resolution, we need prior information about the modeling pair $(\Theta,\mathcal{P})$, especially the model map. For natural images, this is hard even for simple objects. The deep learning-based SISR resolves this issue by approximating the resolution-enhancing map $\mathcal{L}$ via deep hidden layers.

For generative models, one of the key concepts is the latent space, which is used to represent the signal. The representation captures the intrinsic structure of the signal which lies in a high-dimensional space. The crucial step is to learn the map from the latent space to the signal space. The Model-SR has a similar structure. The parameter space $\Theta$ acts as an analogy to the latent space and $\mathcal{M}$ is signal space that is embedded in high (even infinite) dimensional space. The model map $\mathcal{P}$ is the bridge between the latent and signal space. Thus, the learning step in generative models can be interpreted as finding (an approximation of) a model map for the signal space. 

The relationship between Model-SR and deep learning is bidirectional. In this paper, we present a theoretical framework for the super-resolution problem by intrinsically modeling the signal space. Such a framework guarantees a stable extrapolation of high-frequency data from the low-frequency measurement in the frequency domain. We expect that the theoretical foundations of Model-SR will offer valuable insights and serve as a starting point for understanding the theory behind neural network-based generative models. Given the difficulty in identifying the appropriate modeling pair for the signal space, the methods used in neural network-based generative models can be instrumental in discovering efficient representations and modeling strategies within the Model-SR framework. We believe that integrating the techniques of neural networks with Model-SR will provide a powerful tool for solving inverse problems with similar structural characteristics.

\section{Conclusion}\label{sec: discussion}
In this paper, we develop the theory of the model-based super-resolution framework. We present the general mathematical theory along with concrete examples and numerical experiments. We show that under suitable modeling, super-resolution problems enjoy certain stability. 

Within the proposed framework, the challenging part is the non-convex nature of the objective function, for which good initial guesses are needed for the convergence of gradient descent algorithms. Efficient methods for selecting good initial guesses in each concrete model shall be studied, and we leave it as future work. The model-based framework can be generalized to other problems. We expect that the results shown in this paper offer another perspective on the super-resolution problem and take a step forward in understanding model-based problems and problems having a similar structure.
\section{Appendix}
\subsection{Proof of Theorem \ref{thm: abstract stability}}\label{prf: general stability}

If $U$ is a singleton, the result is trivial. We assume that $\operatorname{diam} U := \sup \{ \|\theta-\theta'\|: \theta,\theta'\in U\}>0$, and it is clear that $\operatorname{diam} U<\infty$ since $U$ is compact.\\

    \textbf{Step 1.} Large distance case\\
    
    For any given $r>0$, we consider the set $S := \{ (\theta,\theta')\in U \times U : \|\theta-\theta'\|\ge r\}$. If $S$ is empty, then it is trivial. Otherwise, notice that $S$ is compact and the map $(\theta,\theta')\mapsto \|\PL(\theta)-\PL(\theta')\|$ is continuous, we can then define
    \[
    C_1 := \min \{ \|\PL(\theta)-\PL(\theta')\|: (\theta,\theta')\in S \}.
    \]
    The injectivity of $\PL$ guarantees that $C_1'>0$, and we have 
    \begin{align}{\label{abs_thm_res_1}}
        \|\theta-\theta'\| \le C_{U} \cdot \|\PL(\theta)-\PL(\theta')\|, 
    \end{align}
    with $C_{U} = \frac{\operatorname{diam} U}{C_1}$.\\
    
   \textbf{Step 2.} Short distance case\\
   
   We assume that $\theta,\theta'\in U$ satisfying $\|\theta-\theta'\| \le r_U$ for some $r_U>0$ to be determined later. Let $\gamma(t)$ be the line segment defined by $\gamma(t) = (1-t)\theta+t\theta'\in U, \ t\in [0,1]$. By the convexity of $U$, we have $\gamma(t)\in U$ for all $t\in [0,1]$. Combining the fundamental theorem of calculus and $\PL \in C^1(\Real^m,\HL)$, we can write 
   \[
   \PL(\theta')-\PL(\theta) = \int_0^1 D(\PL\circ\gamma)(t) dt = \int_0^1 D\PL(\gamma(t))(\theta'-\theta) dt.
   \]
   Therefore,
   $$
D\PL(\theta)(\theta-\theta')=\PL(\theta)-\PL(\theta')+\int_0^1\left[D\PL(\theta)-D\PL(\gamma(t))\right](\theta-\theta') d t,
$$
and further we have
$$
\left\|D\PL(\theta)(\theta-\theta')\right\| \leqslant\|\PL(\theta)-\PL(\theta')\|+\int_0^1\left\|D\PL(\theta)-D\PL(\gamma(t))\right\|_{op}\cdot\|\theta-\theta'\| d t .
$$
By rearrangement and straightforward estimation, we have
\begin{align}{\label{abs_thm_est1}}
\frac{\|\PL(\theta)-\PL(\theta')\|}{\|\theta-\theta'\|} \ge \inf _{z \in \mathbb{S}^{m-1}}\left\{\left\|D\PL(\theta) z\right\|\right\}-\sup _{t \in[0,1]}\left\|D\PL(\theta)-D\PL(\gamma(t))\right\|_{op},
\end{align}
which holds for any $\theta, \theta' \in U$ satisfying $\|\theta-\theta'\|\le r_U$, where $\mathbb{S}^{m-1}$ is the unit sphere. 
We then show the right-hand side of (\ref{abs_thm_est1}) is bounded below away from 0.\\

The injectivity of $D\PL(\theta)$ in $U$ as well as the compactness of $U$ and $\mathbb{S}^{m-1}$ yield that
$$
C_2:= \inf _{\theta \in U, z \in \mathbb{S}^{m-1}}\left\|D\PL(\theta) z\right\|>0 .
$$

On the other hand, since $\PL \in C^1(\Real^m, \HL)$, $U$ is compact, and $\gamma(t) \in U$ for any $t \in[0,1]$, there exists a non-decreasing modulus of continuity $\omega_{D\PL, U}$ such that
$$
\left\|D\PL(\theta)-D\PL(\gamma(t))\right\|_{op} \leqslant \omega_{D\PL, U}\left(\|\theta-\gamma(t)\|\right) \leqslant \omega_{D\PL, U}\left(\|\theta-\theta'\|\right) \leqslant \omega_{D\PL, U}(r_U),
$$
for every $t \in[0,1]$. \\

Then, by choosing a sufficient small $r_U>0$ such that $\omega_{D\PL, U}(r) \leqslant \frac{C_2}{2}$, we have 
\begin{align}
        \|\theta-\theta'\| \le C_{U} \cdot \|\PL(\theta)-\PL(\theta')\|, 
\end{align}
with $C_U = \frac{2}{C_2}$.\\

\textbf{Step 3.} \\

Using the convexity of $U$, it is straightforward that
\begin{align}
    \| \PH(\theta)-\PH(\theta')\| \le \|D\PH\|_{op}\cdot \| \theta-\theta' \| \le C_{U}\cdot \|D\PH\|_{op} \cdot\| \PL(\theta)-\PL(\theta')\|.
\end{align}

\subsection{Proof of Theorem \ref{thm: abs_convergence}}
\textbf{Step 1.}\\

let $\hat{\theta}\in U$ be a solution to $(\ref{opt_problem})$, we have 
\begin{align}
    \| \PL(\hat{\theta})-y \|\le \|\PL(\theta^*)-y \| = \| W\|<\sigma.
\end{align}
Hence, $\hat{\theta}$ is an admissible solution.\\

\textbf{Step 2.}\\

Straightforward calculation gives
\begin{align*}
    &\nabla \varphi(\theta) = \Re\{ \overline{D\PL}^\top \bra{\PL(\theta)-y}\},\\
    &\nabla^2 \varphi(\theta) = \Re\{\overline{D\PL}^\top D\PL+ \mathcal{N}(\theta)\},
\end{align*}
where $\mathcal{N}(\theta)=\sum_{k=-K_L}^{K_L} \bra{\mathcal{P}_{L,k}(\theta)-y_k}\cdot \nabla^2 \overline{\mathcal{P}}_{L,k}$.\\

By the assumption that $\PL\in C^2(\Real^m,\HL)$ and $U$ is compact, there exists $A>0$, s.t. $\|\overline{D\PL}^\top D\PL\|\le A$ and $\|\nabla^2 \mathcal{P}_{L,k} \|_{op}\le A$, for $k=-K_L,\cdots,K_L$. Then, it is clear that there exists $\nu_u>0$, s.t. $\nabla^2\varphi(\theta) \preccurlyeq \nu_u \mathcal{I}$, $\forall \theta\in U$.\\

By the assumption that $D\PL(\theta)$ is injective for all $\theta\in U$, the matrix $\overline{D\PL(\hat{\theta})}^\top D\PL(\hat{\theta})$ is positive definite, which implies $\sigma_{\min}\bra{D\PL(\hat{\theta})}>0$. Let $\xi = (\|\nabla^2\mathcal{P}_{L,-K_L}\|_{op},\cdots,\|\nabla^2\mathcal{P}_{L,K_L}\|_{op})$. Notice that,
\begin{align}
    \|\mathcal{N}(\hat{\theta})\|_{op} 
    &\le \sum_{k=-K_L}^{K_L} \left| \mathcal{P}_{L,k}(\hat \theta)-y_k \right| \| \nabla^2 \overline{\mathcal{P}}_{L,k}\|_{op}\notag \\
    &\le \bra{\sum_{k=-K_L}^{K_L} \left| \mathcal{P}_{L,k}(\hat \theta)-y_k \right|^2 }^{1/2} \cdot \| \xi\| \notag \\
    & < \|\xi\|\cdot \sigma.
\end{align}
By Weyl's theorem, we have
\begin{align*}
    \lambda_{\min}\bra{\overline{D\PL(\hat{\theta})}^\top D\PL(\hat{\theta})+\mathcal{N}(\hat{\theta})}
    &\ge \lambda_{\min}\bra{\overline{D\PL(\hat{\theta})}^\top D\PL(\hat{\theta})}-\|\mathcal{N}(\hat \theta)\|_{op} \\
    &> \sigma^2_{\min}\bra{D\PL(\hat{\theta})}-\|\xi\|\cdot \sigma\ge 0.
\end{align*}
Therefore, the matrix $\overline{D\PL(\hat{\theta})}^\top D\PL(\hat{\theta})+\mathcal{N}(\hat{\theta})$ is positive definite and so is $\nabla^2 \varphi(\hat{\theta})$. Since $\nabla^2 \varphi(\theta)$ depends continuously on $\theta$, there exists a closed neighborhood $U_{\hat{\theta}}\subset U$, such that $\nabla^2 \varphi(\theta)$ is positive definite for all $\theta\in U_{\hat{\theta}}$. The proof is finished by taking $$\nu_l = \inf_{\theta\in U_{\hat{\theta}}} \lambda_{\min}\bra{\nabla^2\varphi(\theta)}.$$

\subsection{Proof of Theorem \ref{thm: stability point source}}
We notice that $U=[-\frac{n-1}{4K},\frac{n-1}{4K}]$ is compact and convex. According to Theorem \ref{thm: abstract stability}, we only need to verify $\PL|_{U}$ is injective and $D\PL(\theta)$ is injective for all $\theta \in U$.\\
\textbf{Step 1.} Injectivity of $\PL|_{U}$: \\

For $\theta,\theta' \in U$, we assume that $\PL(\theta) = \PL(\theta')$, i.e. 
\begin{align}{\label{injective_eq_1}}
    \sum_{j=1}^n \theta_{j,1}e^{-2\pi i\theta_{j,2}k} = \sum_{j=1}^n \theta_{j,1}'e^{-2\pi i\theta_{j,2}'k} \quad |k|\le K_L.
\end{align}
Without loss of generality, we suppose $\theta_{1,2} = \theta_{1,2}', \cdots, \theta_{s,2} = \theta_{s,2}'$. 

Denote $$\phi_{K_L}(\theta_{j,2}) = e^{2\pi i\theta_{j,2}K_L}\cdot\bra{1,e^{-2\pi i\theta_{j,2}},\cdots,e^{-4\pi i\theta_{j,2}K_L}}^\top,$$ and define $$A_{\theta} = \bra{ \phi_{K_L}(\theta_{1,2}),\cdots,\phi_{K_L}(\theta_{s,2}),\phi_{K_L}(\theta_{s+1,2}),\cdots,\phi_{K_L}(\theta_{n,2}),\phi_{K_L}(\theta'_{s+1,2}),\cdots,\phi_{K_L}(\theta'_{n,2})},$$ $$\phi_{\theta} = (\theta_{1,1}-\theta_{1,1}',\cdots,\theta_{s,1}-\theta_{s,1}',\theta_{s+1,1},\cdots,\theta_{n,1},\theta_{s+1,1}',\cdots,\theta_{n,1}')^\top$$. We rewrite the equations (\ref{injective_eq_1}) in the following matrix form
\begin{align*}
    A_{\theta}\cdot \phi_{\theta} = 0.
\end{align*}
It is easy to verify that $A_{\theta}$ has full column rank, we deduce that $\theta_{s+1} =\cdots= \theta_n =\theta_{s+1}' = \cdots = \theta_n'  =0$, which contradicts to the assumption that $\theta_j\ne 0$. Thus, the only case in which there is no contradiction is $\theta_{2,j} = \theta_{2,j}'$ for all $j = 1,\cdots,n$. Then we rewrite the following equations into a matrix form 
\begin{align}
    \sum_{j=1}^n \theta_{j,1}e^{-2\pi i\theta_{j,2}k} = \sum_{j=1}^n \theta_{j,1}'e^{-2\pi i\theta_{j,2}k} \quad |k|\le K_L.
\end{align}
We derive that 
\begin{align}
     \tilde{A}_\theta \cdot \tilde{\phi_\theta} = 0, 
\end{align}
where $\tilde{A}_\theta = \bra{ \phi_{K_L}(\theta_{1,2}),\cdots,\phi_{K_L}(\theta_{1,2})}$ and $\tilde{\phi_\theta} = (\theta_{1,1}-\theta_{1,1}',\cdots,\theta_{n,1}-\theta_{n,1}')^\top$. Since $\tilde{A}_\theta$ has full column rank, we have $\tilde{\phi_\theta}=0$, which implies $\theta_{1,j} = \theta_{1,j}'$ for all $j = 1,\cdots,n$. Therefore, $\PL|_{U}$ is injective.\\

\textbf{Step 2.} Injectivity of $D\PL|_{U}$:\\

We calculate that 
\begin{align}
    &\frac{\partial \mathcal{P}_{L,k}}{\partial \theta_{j,1}} = e^{-2\pi i\theta_{j,2}k}, \\
    &\frac{\partial \mathcal{P}_{L,k}}{\partial \theta_{j,2}} = -2\pi i \cdot \theta_{j,1}k e^{-2\pi i\theta_{j,2}k}, 
\end{align}
and
\begin{align}
    D\PL(\theta) = \bra{
    \begin{matrix}
        \frac{\partial \mathcal{P}_{L,-K_L}}{\partial \theta_{1,1}}&\cdots&\frac{\partial \mathcal{P}_{L,-K_L}}{\partial \theta_{n,1}}&\frac{\partial \mathcal{P}_{L,-K_L}}{\partial \theta_{1,2}}&\cdots&\frac{\partial \mathcal{P}_{L,-K_L}}{\partial \theta_{n,2}}\\
        \vdots& &\vdots&\vdots& &\vdots\\
        \frac{\partial \mathcal{P}_{L,K_L}}{\partial \theta_{1,1}}&\cdots&\frac{\partial \mathcal{P}_{L,K_L}}{\partial \theta_{n,1}}&\frac{\partial \mathcal{P}_{L,K_L}}{\partial \theta_{1,2}}&\cdots&\frac{\partial \mathcal{P}_{L,K_L}}{\partial \theta_{n,2}}\\
    \end{matrix}
    }.
\end{align}
For any $\theta\in U$, the confluent Vandermonde matrix $D\PL$ has full column rank and thus $D\PL|_{U}$ is injective.\\

\textbf{Step 3.} The inequality $\|\hat \theta-\theta\|<C_1\cdot \sigma$ is a consequence of Theorem 3.8 and Theorem 3.10 in \cite{liu2024mathematical}. To get the result, one only needs to notice that the cutoff frequency and the difference of the Fourier transform adopted in \cite{liu2024mathematical} imply the Rayleigh limit is $\frac{\pi}{\Omega}$, which is $\frac{1}{2K_L}$ in this paper.

\textbf{Step 4.} Estimation of $\|D\PH(\theta)\|_{op}$:\\

Straightforward calculation gives that 
\begin{align}
    D\PH(\theta) = \bra{
    \begin{matrix}
        \frac{\partial \mathcal{P}_{H,-K_H}}{\partial \theta_{1,1}}&\cdots&\frac{\partial \mathcal{P}_{H,-K_H}}{\partial \theta_{n,1}}&\frac{\partial \mathcal{P}_{H,-K_H}}{\partial \theta_{1,2}}&\cdots&\frac{\partial \mathcal{P}_{H,-K_H}}{\partial \theta_{n,2}}\\
        \vdots& &\vdots&\vdots& &\vdots\\
        \frac{\partial \mathcal{P}_{H,K_H}}{\partial \theta_{1,1}}&\cdots&\frac{\partial \mathcal{P}_{H,K_H}}{\partial \theta_{n,1}}&\frac{\partial \mathcal{P}_{H,K_H}}{\partial \theta_{1,2}}&\cdots&\frac{\partial \mathcal{P}_{H,K_H}}{\partial \theta_{n,2}}\\
    \end{matrix}
    }.
\end{align}
For any given $\theta \in U$, we have
\begin{align}
    \|D\PH(\theta)\|_{op}^2 \le \|D\PH(\theta)\|_F^2 &= (2K_H+1)n+\frac{4\pi^2}{3}\sum_{j=1}^n \theta_{j,1}^2\cdot\bra{K_H(K_H+1)(2K_H+1)} \notag\\
    &\le (2K_H+1)n+\frac{4n\pi^2A_I^2}{3}\cdot\bra{K_H(K_H+1)(2K_H+1)}.
\end{align}

Meanwhile, by the definition of operator norm, we have
\begin{align}
    \|D\PH(\theta) \|_{op} \ge \| D\PH(\theta)\alpha_{l} \|, \quad l = -K_H,\cdots,K_H,
\end{align}
where $\{\alpha_l\}_{l=-K_H}^{K_H}$ is the canonical basis of $\Real^{2n}$. Then, the straightforward calculation gives
\begin{align}
    \|D\PH(\theta) \|_{op} \gtrsim K_H^{3/2}.
\end{align}

\subsection{Proof of Theorem \ref{thm: stability diff Dirac}}
We notice that $U$ is compact and convex. According to Theorem \ref{thm: abstract stability},, we only need to verify $\PL|_{U}$ is injective and $D\PL|_{U}(\theta)$ is injective for all $\theta \in U$.\\

\textbf{Step 1.} Injectivity of $\PL|_{U}$: \\

For $\theta,\theta'\in U$, we assume that $\PL(\theta) = \PL(\theta')$, i.e.
\begin{align}
     \sum_{j=1}^n \sum_{r=0}^{R_j} \theta_{r,j,1}\cdot (-2\pi i k)^r  e^{-2\pi i \theta_{j,2}k} 
    =  \sum_{j=1}^n \sum_{r=0}^{R_j} \theta_{r,j,1}'\cdot (-2\pi i k)^r  e^{-2\pi i \theta_{j,2}'k}, \quad |k|\le K_L.
\end{align}
By the similar method of the proof of Theorem \ref{thm: stability point source}, we only need to show that sources with different orders generate independent signals and the same order source with different positions generate independent signals. It then suffices to show that for $\theta_{j,2} \ne \theta'_{j,2}$, the matrix 
\[
\bra{
\begin{matrix}
    e^{2\pi i \theta_{j,2} K_L} & \cdots & (2\pi i K_L)^{R_1}e^{2\pi i \theta_{j,2} K_L} & e^{2\pi i \theta'_{j,r,2} K_L} & \cdots & (2\pi i K_L)^{R_1}e^{2\pi i \theta'_{j,r,2} K_L}\\
    \vdots & &\vdots&\vdots& & \vdots\\
    e^{-2\pi i \theta_{j,2} K_L} & \cdots & (-2\pi i K_L)^{R_2}e^{-2\pi i \theta_{j,2} K_L} & e^{-2\pi i \theta'_{j,r,2} K_L} & \cdots & (-2\pi i K_L)^{R_2}e^{-2\pi i \theta'_{j,r,2} K_L}
\end{matrix}
}
\]
has full column rank for any $R_1$ and $R_2$, and it is straightforward by its confluent Vandermonde matrix structure. Therefore, $\PL|_{U}$ is injective.\\

\textbf{Step 2.} Injectivity of $D\PL|_{U}$:\\

We calculate that 
\begin{align}
    &\frac{\partial \mathcal{P}_{L,k}}{\partial \theta_{r,j,1}}= (-2\pi i k)^r\cdot e^{-2\pi i \theta_{j,2}k},\\
    &\frac{\partial \mathcal{P}_{L,k}}{\partial \theta_{j,2}} = \theta_{r,j,1}\cdot(-2\pi i k)^{r+1}\cdot e^{-2\pi i \theta_{j,2}k},
\end{align}
and 
\begin{align*}
    D\PL(\theta) = \bra{
    \begin{matrix}
        \frac{\partial  \mathcal{P}_{L,-K_L}}{\partial \theta_{0,1,1}}&\cdots&\frac{\partial  \mathcal{P}_{L,-K_L}}{\partial \theta_{R,n_R,1}}&\frac{\partial  \mathcal{P}_{L,-K_L}}{\partial \theta_{0,1,2}}&\cdots&\frac{\partial  \mathcal{P}_{L,-K_L}}{\partial \theta_{R,n_R,2}}\\
        \vdots& &\vdots&\vdots & &\vdots\\
        \frac{\partial  \mathcal{P}_{L,K_L}}{\partial \theta_{0,1,1}}&\cdots&\frac{\partial  \mathcal{P}_{L,K_L}}{\partial \theta_{R,n_R,1}}&\frac{\partial  \mathcal{P}_{L,K_L}}{\partial \theta_{0,1,2}}&\cdots&\frac{\partial  \mathcal{P}_{L,K_L}}{\partial \theta_{R,n_R,2}}
    \end{matrix}
    }.
\end{align*}
The confluent Vandermonde matrix structure of $D\PL$ implies its injectivity for all $\theta\in U$.\\

\textbf{Step 3.} Estimation of $\|D\PH(\theta)\|_{op}$:\\

It is clear that

\begin{align*}
    D\PH(\theta) = \bra{
    \begin{matrix}
        \frac{\partial  \mathcal{P}_{L,-K_H}}{\partial \theta_{0,1,1}}&\cdots&\frac{\partial  \mathcal{P}_{L,-K_H}}{\partial \theta_{R,n_R,1}}&\frac{\partial  \mathcal{P}_{L,-K_H}}{\partial \theta_{0,1,2}}&\cdots&\frac{\partial  \mathcal{P}_{L,-K_H}}{\partial \theta_{R,n_R,2}}\\
        \vdots& &\vdots&\vdots & &\vdots\\
        \frac{\partial  \mathcal{P}_{L,K_H}}{\partial \theta_{0,1,1}}&\cdots&\frac{\partial  \mathcal{P}_{L,K_H}}{\partial \theta_{R,n_R,1}}&\frac{\partial  \mathcal{P}_{L,K_H}}{\partial \theta_{0,1,2}}&\cdots&\frac{\partial  \mathcal{P}_{L,K_H}}{\partial \theta_{R,n_R,2}}
    \end{matrix}
    }.
\end{align*}

\normalsize
For any given $\theta \in U$, we have 
\begin{align*}
    \|D\PH \|_{op}^2 \le \|D\PH \|_F^2 \le \sum_{k=-K_H}^{K_H} \sum_{r=0}^R n_r(2\pi k)^{2r}\bra{1+4\pi^2k^2A_I^2}.
\end{align*}
Meanwhile, by the definition of operator norm, we have
\begin{align}
    \|D\PH(\theta) \|_{op} \ge \| D\PH(\theta)\alpha_{l} \|, \quad l = -K_H,\cdots,K_H,
\end{align}
where $\{\alpha_l\}_{l=-K_H}^{K_H}$ is the canonical basis of $\Real^{2N}$. Then, the straightforward calculation gives
\begin{align}
    \|D\PH(\theta) \|_{op}\gtrsim \sqrt{\sum_{k=-K_H}^{K_H} k^{2R+2}} \gtrsim K_H^{R+3/2}.
\end{align}

\subsection{Proof of Theorem \ref{thm: stability gauss}}
We notice that $U$ is compact and convex. According to Theorem \ref{thm: abstract stability}, we only need to verify $\PL|_{U}$ is injective and $D\PL(\theta)$ is injective for all $\theta \in U$.\\

\textbf{Step 1.} Injectivity of $\PL|_U$:\\

It suffices to show that for $\theta_{j,2}\ne\theta_{j,2}'$ for $j = 1,\cdots,n$, we have the following matrix has full column rank:
\begin{align*}
    B_{\alpha}\left(
    \begin{array}{cc}
        M_{\theta} & M_{\theta'}
    \end{array}
    \right),
\end{align*}
where  
\begin{align*}
    B_\alpha = \operatorname{diag}\bra{e^{-2\pi^2\alpha^2\omega_{-K}^2},\cdots, e^{-2\pi^2\alpha^2\omega_{-K}^2}},
\end{align*}
and 
\begin{align*}
    M_\theta = \left(
    \begin{array}{ccc}
     e^{-2\pi i\theta_{1,2}\omega_{-K}} & \cdots & e^{-2\pi i\theta_{n,2}\omega_{-K}} \\
     \vdots & \ddots & \vdots \\
     e^{-2\pi i\theta_{1,2}\omega_{-K}} & \cdots & e^{-2\pi i\theta_{n,2}\omega_{-K}}
    \end{array}\right), \quad
    M_{\theta'} = \left(
    \begin{array}{ccc}
     e^{-2\pi i\theta'_{1,2}\omega_{-K}} & \cdots & e^{-2\pi i\theta'_{n,2}\omega_{-K}} \\
     \vdots & \ddots & \vdots \\
     e^{-2\pi i\theta'_{1,2}\omega_{-K}} & \cdots & e^{-2\pi i\theta'_{n,2}\omega_{-K}}
    \end{array}\right).
\end{align*}
Indeed, this is easy to check by noticing the Vandermonde structure and $B_\alpha$ is full rank.

\textbf{Step 2.} Injectivity of $D\PL|_U$:

We calculate that 
\begin{align}
    &\frac{\partial \mathcal{P}_{L,k}}{\partial_{j,1}} = \sqrt{2\pi\alpha^2}\cdot e^{-2\pi^2\alpha^2\omega_{k}^2}\cdot e^{-2\pi i\theta_{1,2}\omega_{k}}\\
    &\frac{\partial \mathcal{P}_{L,k}}{\partial_{j,2}} = \sqrt{2\pi\alpha^2}\cdot e^{-2\pi^2\alpha^2\omega_{k}^2}\cdot \theta_{j,1}\cdot (-2\pi i\omega_k)e^{-2\pi i\theta_{1,2}\omega_{k}}
\end{align}
We can apply the similar argument as step 1 to write the corresponding matrix, and the injectivity is the consequence of the confluent Vandermonde structure.

\textbf{Step 3.} Estimation of $\|D\PH(\theta)\|_{op}$:\\

It is clear that 
\begin{align}
    D\PH(\theta) = \bra{
    \begin{matrix}
        \frac{\partial \mathcal{P}_{H,-K_H}}{\partial \theta_{1,1}}&\cdots&\frac{\partial \mathcal{P}_{H,-K_H}}{\partial \theta_{n,1}}&\frac{\partial \mathcal{P}_{H,-K_H}}{\partial \theta_{1,2}}&\cdots&\frac{\partial \mathcal{P}_{H,-K_H}}{\partial \theta_{n,2}}\\
        \vdots& &\vdots&\vdots& &\vdots\\
        \frac{\partial \mathcal{P}_{H,K_H}}{\partial \theta_{1,1}}&\cdots&\frac{\partial \mathcal{P}_{H,K_H}}{\partial \theta_{n,1}}&\frac{\partial \mathcal{P}_{H,K_H}}{\partial \theta_{1,2}}&\cdots&\frac{\partial \mathcal{P}_{H,K_H}}{\partial \theta_{n,2}}\\
    \end{matrix}
    }.
\end{align}

For any given $\theta\in U$, we have
\begin{align*}
    \|D\PH \|_{op}^2 \le \|D\PH \|_F^2 
    \le 2\pi \alpha \sum_{j=1}^n \sum_{k=-K_H}^{K_H}\bra{1+4\pi k^2\theta_{j,1}^2}e^{-2\pi^2\alpha^2k^2}  \triangleq C_3,
\end{align*}
where $C_3$ is independent of $K_H$ due to the convergence of the series.

\subsection{Proof of Proposition \ref{prop4-4}}
The signal can be written as 
\[
y_k = (2\pi i k)^r e^{-2\pi i k z}+(2\pi i k)^r e^{-2\pi i k z'},
\]
We calculate that 
\[
\frac{\partial \mathcal{P}_{l,k}}{\partial z} = (2\pi i k)^{r+1}e^{-2\pi i k z}, \quad \frac{\partial \mathcal{P}_{l,k}}{\partial z'} = (2\pi i k)^{r+1}e^{-2\pi i k z'}.
\]
Then, the Jacobi matrix can be written as 
\[
D\PL = \bra{
\begin{array}{cc}
    (-2\pi i K_L)^{r+1}e^{2\pi i K_L z}, & (-2\pi i K_L)^{r+1}e^{2\pi i K_L z'}\\
     \vdots&\vdots \\
     (2\pi i K_L)^{r+1}e^{-2\pi i K_L z}, & (2\pi i K_L)^{r+1}e^{-2\pi i K_L z'}
\end{array}
}
\]

We calculate its normal matrix as 
\[
(D\PL)^*D\PL = \bra{
\begin{array}{cc}
    \sum_{k=-K_L}^{K_L}(2\pi k)^{2r+2}, & \sum_{k=-K_L}^{K_L}(2\pi k)^{2r+2} e^{-2\pi i k (z'-z)}\\
     \sum_{k=-K_L}^{K_L}(2\pi k)^{2r+2} e^{-2\pi i k (z-z')}, & \sum_{k=-K_L}^{K_L}(2\pi k)^{2r+2}
\end{array}
}.
\]
Thus, the minimum singular value of $D\PL$ can be estimated as follows:
\begin{align}
    \sigma_{\min}^2 &= \sum_{k=-K_L}^{K_L}(2\pi k)^{2r+2} - \left|\sum_{k=-K_L}^{K_L}(2\pi k)^{2r+2} e^{-2\pi i k (z'-z)}\right| \notag \\
    & = 2 \sum_{k=1}^{K_L}(2\pi k)^{2r+2} - 2 \left|\sum_{k=1}^{K_L}(2\pi k)^{2r+2} \cos(2\pi  k (z'-z))\right| \notag\\
    & = 2 \sum_{k=1}^{K_L}(2\pi k)^{2r+2} \bra{1-\cos(2\pi  k |z'-z|)}\notag \\
    & \ge 2 \sum_{k=1}^{K_L}(2\pi k)^{2r+2} (8k^2\Delta^2) \quad \bra{ \text{by } 1-\cos x \ge \frac{2}{\pi^2} x^2, x \in (0,\pi)} \notag \\
    & \ge C_1K_L^{2r+5}\Delta^2
\end{align}
for some constants $C_1$.
We can therefore conclude that 
\begin{align}
    \|\theta -\theta' \| \le \frac{1}{\sigma_{\min}} \|\PL(\theta)-\PL(\theta') \| \le \frac{C}{ K_L^{r+5/2}\Delta}\|\PL(\theta)-\PL(\theta') \|,
\end{align}
for some constant $C>0$.

\subsection{Proof of Propostion \ref{prop4-5}}
The signal can be written as 
\[
y_k = e^{-2\pi \alpha^2 k^2} \bra{ e^{-2\pi i k z}+e^{-2\pi i k z'}},
\]
We calculate that 
\[
\frac{\partial \mathcal{P}_{l,k}}{\partial z} = e^{-2\pi \alpha^2 k^2}(-2\pi i k)e^{-2\pi i k z}, \quad \frac{\partial \mathcal{P}_{l,k}}{\partial z'} = e^{-2\pi \alpha^2 k^2}(-2\pi i k)e^{-2\pi i k z'}.
\]
Then, the Jacobi matrix can be written as 
\[
D\PL = \bra{
\begin{array}{cc}
    e^{-2\pi \alpha^2 K_L^2}(2\pi i K_L)e^{2\pi i K_L z}, & e^{-2\pi \alpha^2 K_L^2}(2\pi i K_L)e^{-2\pi i K_L z}\\
     \vdots&\vdots \\
     e^{-2\pi \alpha^2 K_L^2}(-2\pi i K_L)e^{-2\pi i K_L z}, & e^{-2\pi \alpha^2 K_L^2}(-2\pi i K_L)e^{-2\pi i K_L z}
\end{array}
}
\]

We calculate its normal matrix as 
\[
(D\PL)^*D\PL = \bra{
\begin{array}{cc}
    \sum_{k=-K_L}^{K_L}(2\pi k)^{2}e^{-4\pi \alpha^2 k^2}, & \sum_{k=-K_L}^{K_L}(2\pi k)^{2}e^{-4\pi \alpha^2 k^2} e^{-2\pi i k (z'-z)}\\
     \sum_{k=-K_L}^{K_L}(2\pi k)^{2} e^{-2\pi i k (z-z')}, & \sum_{k=-K_L}^{K_L}(2\pi k)^{2}e^{-4\pi \alpha^2 k^2}
\end{array}
}.
\]
Thus, the minimum singular value of $D\PL$ can be estimated as follows:
\begin{align}
    \sigma_{\min}^2 &= \sum_{k=-K_L}^{K_L}(2\pi k)^{2}e^{-4\pi \alpha^2 k^2} - \left|\sum_{k=-K_L}^{K_L}(2\pi k)^{2}e^{-4\pi \alpha^2 k^2} e^{-2\pi i k (z'-z)}\right| \notag \\
    & = 2 \sum_{k=1}^{K_L}(2\pi k)^{2}e^{-4\pi \alpha^2 k^2} - 2 \left|\sum_{k=1}^{K_L}(2\pi k)^{2}e^{-4\pi \alpha^2 k^2} \cos\bra{2\pi k (z'-z)}\right| \notag\\
    & = 2 \sum_{k=1}^{K_L}(2\pi k)^{2}e^{-4\pi \alpha^2 k^2} \bra{1-\cos(2\pi  k |z'-z|)}\notag \\
    & \ge 2 \sum_{k=1}^{K_L}(2\pi k)^{2}e^{-4\pi \alpha^2 k^2} (8k^2\Delta^2) \quad \bra{ \text{by } 1-\cos x \ge \frac{2}{\pi^2} x^2, x \in (0,\pi)} \notag \\
    & \ge C_1(K_L)\Delta^2
\end{align}
for some constants $C_1(K_L)$.
We can therefore conclude that 
\begin{align}
    \|\theta -\theta' \| \le \frac{1}{\sigma_{\min}} \|\PL(\theta)-\PL(\theta') \| \le \frac{1}{ C(K_L)\Delta}\|\PL(\theta)-\PL(\theta') \|,
\end{align}
for some constant $C(K_L)>0$ with $C(K_L)\rightarrow C$ as $K_L\rightarrow \infty$.

\newpage
\bibliographystyle{ieeetr}
\bibliography{references} 
\newpage

\end{document}